\documentclass[a4paper,UKenglish,cleveref, autoref, thm-restate]{lipics-v2021}

\hideLIPIcs

\title{An Efficient Algorithm for All-Pairs Bounded Edge Connectivity}

\theoremstyle{definition}
\newtheorem{problem}[theorem]{Open Problem}
\newtheorem*{problem*}{Open Problem}

\newtheorem*{conj*}{Conjecture}

\DeclareMathOperator{\adj}{adj}
\renewcommand{\t}{\text}

\newcommand{\indeg}{\deg_{\text{in}}}
\newcommand{\outdeg}{\deg_{\text{out}}}
\newcommand{\grp}[1]{\left(#1\right)}
\newcommand{\set}[1]{\left\{#1\right\}}

\DeclareMathOperator{\rank}{rank}
\DeclareMathOperator{\poly}{poly}
\newcommand{\outedges}[1]{E_{\textup{out}}(#1)}
\newcommand{\inedges}[1]{E_{\textup{in}}(#1)}
\newcommand{\outnodes}[1]{V_{\textup{out}}(#1)}
\newcommand{\outnodescl}[1]{V_{\textup{out}}[#1]}
\newcommand{\innodes}[1]{V_{\textup{in}}(#1)}
\newcommand{\innodescl}[1]{V_{\textup{in}}[#1]}
\newcommand{\nnew}{n_{\textup{new}}}
\newcommand{\mnew}{m_{\textup{new}}}
\newcommand{\sub}{\subseteq}

\usepackage[ruled]{algorithm}
\usepackage[noend]{algpseudocode}

\author{Shyan Akmal}{MIT EECS and CSAIL, Cambridge MA, USA \and \url{https://www.shyanakmal.com}}{naysh@mit.edu}{https://orcid.org/0000-0002-7266-2041}{Supported in part by NSF grants CCF-2129139 and CCF-2127597.}

\author{Ce Jin}{MIT EECS and CSAIL, Cambridge MA, USA \and \url{https://ce-jin.github.io/}}{cejin@mit.edu}{}{Partially supported by NSF Grant CCF-2129139 and a Siebel Scholarship.}

\authorrunning{S. Akmal and C. Jin} 

\Copyright{Shyan Akmal and Ce Jin} 



\category{} 

\relatedversion{}

\acknowledgements{The first author thanks Virginia Vassilevska Williams for insightful discussions on algorithms for computing matrix rank. }

\nolinenumbers 

\begin{document}

\maketitle

\begin{abstract}
Our work concerns algorithms for a variant of \textsf{Maximum Flow} in unweighted graphs.
In the \textsf{All-Pairs Connectivity (APC)} problem, we are given a graph $G$ on $n$ vertices and $m$ edges, and are tasked with computing the maximum number of edge-disjoint paths from $s$ to $t$ (equivalently, the size of a minimum $(s,t)$-cut) in $G$, for all pairs of vertices $(s,t)$.  
Significant algorithmic breakthroughs have recently shown that over undirected graphs, \textsf{APC} can be solved in $n^{2+o(1)}$ time, which is essentially optimal.
In contrast, the true time complexity of \textsf{APC} over directed graphs remains open: this problem can be solved in $\tilde{O}(m^\omega)$ time, where $\omega \in [2, 2.373)$ is the exponent of matrix multiplication, but no matching conditional lower bound is known.

Following 
[Abboud et al., ICALP 2019], 
we study a bounded version of $\textsf{APC}$ called the \textsf{$k$-Bounded All Pairs Connectivity ($k$-APC)} problem.
In this variant of \textsf{APC}, we are given an integer $k$ in addition to the graph $G$, and are now tasked with reporting the size of a minimum $(s,t)$-cut only for pairs $(s,t)$ of vertices with min-cut value less than $k$ (if the minimum $(s,t)$-cut has size at least $k$, we can just report it is ``large'' instead of computing the exact value).

Our main result is an $\tilde{O}((kn)^\omega)$ time algorithm solving \textsf{$k$-APC} in directed graphs.
This is the first algorithm which solves \textsf{$k$-APC} faster than simply solving the more general \textsf{APC} problem exactly, for all $k\ge 3$.
This runtime is $\tilde O(n^\omega)$ for all $k\le \poly(\log n)$, which essentially matches the optimal runtime for the $k=1$ case of \textsf{$k$-APC}, under popular conjectures from fine-grained complexity.
Previously, this runtime was only achieved for $k\le 2$ in general directed graphs
[Georgiadis et al., ICALP 2017], 
and for $k\le o(\sqrt{\log n})$ in the special case of directed acyclic graphs 
[Abboud et al., ICALP 2019]. 
Our result employs the same algebraic framework used in previous work, introduced by 
[Cheung, Lau, and Leung, FOCS 2011].
A direct implementation of this framework involves inverting a large random matrix. 
Our new algorithm is based off the insight that for solving \textsf{$k$-APC}, it suffices to invert a low-rank random matrix instead of a generic random matrix.

We also obtain a new algorithm for a variant of \textsf{$k$-APC}, the \textsf{$k$-Bounded All-Pairs Vertex Connectivity ($k$-APVC)} problem, where we are now tasked with reporting, for every pair of vertices $(s,t)$, the maximum number of internally vertex-disjoint (rather than edge-disjoint) paths from $s$ to $t$ if this number is less than $k$, and otherwise reporting that there are at least $k$ internally vertex-disjoint paths from $s$ to $t$.

Our second result is an $\tilde{O}(k^2n^\omega)$ time algorithm solving \textsf{$k$-APVC} in directed graphs.
Previous work showed how to solve an easier version of the \textsf{$k$-APVC} problem (where answers only need to be returned for pairs of vertices $(s,t)$ which are not edges in the graph) in $\tilde O((kn)^\omega)$ time [Abboud et al, ICALP 2019].
In comparison, our algorithm solves the full \textsf{$k$-APVC} problem, and is faster if $\omega > 2$.
\end{abstract}

\section{Introduction}
\label{sec:intro}

Computing maximum flows is a classic problem which has been extensively studied in graph theory and computer science.
In unweighted graphs, this task specializes to computing connectivities, an interesting computational problem in its own right.
Given a graph $G$ on $n$ vertices and $m$ edges, for any vertices $s$ and $t$ in $G$, the \emph{connectivity} $\lambda(s,t)$ from $s$ to $t$ is defined to be the maximum number of edge-disjoint paths\footnote{By Menger's theorem, $\lambda(s,t)$ is also equal to the minimum number of edges that must be deleted from the graph $G$ to produce a graph with no $s$ to $t$ path.} from $s$ to $t$. 
Since maximum flow can be computed in almost-linear time, we can compute $\lambda(s,t)$ for any given vertices $s$ and $t$ in $m^{1+o(1)}$ time \cite{near-linear-flow}.

What if instead of merely returning the value of a single connectivity, our goal is to compute all connectivities in the graph? 
This brings us to the \textsf{All-Pairs Connectivity (APC)} problem: in this problem, we are given a graph $G$ as above, and are tasked with computing $\lambda(s,t)$ for all pairs of vertices $(s,t)$ in $G$. 
In undirected graphs, \textsf{APC} can be solved in $n^{2+o(1)}$ time \cite{apmf-apsp}, so that this ``all-pairs'' problem is essentially no harder than outputting a single connectivity in dense graphs.

In directed graphs, \textsf{APC} appears to be much harder, with various conditional lower bounds (discussed in \Cref{subsec:related-work}) suggesting it is unlikely this problem can be solved in quadratic time.
Naively computing the connectivity separately for each pair yields an $n^2m^{1+o(1)}$ time algorithm for this problem.
Using the flow vector framework (discussed in \Cref{sec:technical-overview}), it is possible to solve \textsf{APC} in directed graphs in $\tilde{O}(m^\omega)$ time\footnote{Given a function $f$, we write $\tilde{O}(f)$ to denote $f\cdot\poly(\log f).$} \cite{apc}, where $\omega$ is the exponent of matrix multiplication.
Known algorithms imply that $\omega < 2.37286$ \cite{refined-laser-mmult}, so the $\tilde{O}(m^{\omega})$ time algorithm is faster than the naive algorithm whenever the input graph is not too dense. 

Our work focuses on a bounded version of the \textsf{APC} problem, which we formally state as the \textsf{$k$-Bounded All-Pairs Connectivity ($k$-APC)} problem: in this problem, we are given a \emph{directed} graph $G$ as above, and are tasked with computing $\min(k, \lambda(s,t))$ for all pairs  of vertices $(s,t)$ in $G$.
Intuitively, this is a relaxation of the \textsf{APC} problem, where our goal is to compute the exact values of $\lambda(s,t)$ only for pairs $(s,t)$ with small connectivity.
For all other pairs, it suffices to report that the connectivity is large, where $k$ is our threshold for distinguishing between small and large connectivity values. 

When $k=1$, the \textsf{$k$-APC} problem is equivalent to computing the transitive closure of the input graph 
(in this problem, for each pair of vertices $(s,t)$, we are tasked with determining if $G$ contains a path from $s$ to $t$), which can be done in $\tilde{O}(n^{\omega})$ time \cite{transitive-closure-BMM}.
Similarly, for the special case of $k=2$, it is known that \textsf{$k$-APC} can be solved in $\tilde{O}(n^{\omega})$ time, by a divide-and-conquer algorithm employing a cleverly tailored matrix product \cite{2-APC}.
As we discuss in \Cref{subsec:related-work}, there is evidence that these runtimes for \textsf{$k$-APC} when $k\le 2$ are essentially optimal.

Already for $k=3$ however, it is open whether \textsf{$k$-APC} can be solved faster than computing the \emph{exact values} of $\lambda(s,t)$  for all pairs $(s,t)$ of vertices!
Roughly speaking, this is because the known $\tilde{O}(m^\omega)$ time algorithm for \textsf{APC} involves encoding the connectivity information in the inverse of an $m\times m$ matrix, and inverting an $m\times m$ matrix takes $O(m^\omega)$ time in general.
This encoding step appears to be necessary for \textsf{$k$-APC} as well.
For $k = 2$, clever combinatorial observations about the structure of strongly connected graphs allow one to skip this computation, but for $k\ge 3$ it is not clear at all from previous work how to avoid this bottleneck.
Moreover, it is consistent with existing hardness results that \textsf{$k$-APC} could be solved in $O(n^\omega)$ time for any constant $k$. 

\begin{problem}
    \label{prob:1}
    Can \textsf{$k$-APC} be solved in faster than $\tilde{O}(m^{\omega})$ time for $k=3$?
\end{problem}

Due to this lack of knowledge about the complexity of \textsf{$k$-APC}, researchers have also studied easier versions of this problem.
Given vertices $s$ and $t$ in the graph $G$, we  define the \emph{vertex connectivity} $\nu(s,t)$ 
from $s$ to $t$ to be the maximum number of internally vertex-disjoint paths from $s$ to $t$.
We can consider vertex connectivity analogues of the \textsf{APC} and \textsf{$k$-APC} problems.
In the \textsf{All-Pairs Vertex Connectivity (APVC)} problem,  we are given a graph $G$ on $n$ vertices and $m$ edges, and are tasked with computing the value of $\nu(s,t)$ for all pairs of vertices $(s,t)$ in $G$.
In the \textsf{$k$-Bounded All-Pairs Vertex Connectivity ($k$-APVC)} problem, we are given the same input $G$ as  above, but are now tasked with only computing $\min(k,\nu(s,t))$ for all pairs of vertices $(s,t)$ in $G$.

The \textsf{$k$-APVC} problem does not face the $O(m^\omega)$ barrier which existing algorithmic techniques for \textsf{$k$-APC} seem to encounter, intuitively because it is possible to encode all the vertex-connectivity information of a graph in the inverse of an $n\times n$ matrix instead of an $m\times m$ matrix.
As a consequence, \cite{ap-bounded-mincut} was able to present an $\tilde{O}((kn)^\omega)$ time algorithm which computes $\min(k,\nu(s,t))$ for all pairs of vertices $(s,t)$ such that $(s,t)$ is not an edge.
Given this result, it is natural to ask whether the more general \textsf{$k$-APVC} and \textsf{$k$-APC} problems can also be solved in this same running time.

\begin{problem}
    \label{prob:3}
    Can \textsf{$k$-APVC} be solved in $\tilde{O}((kn)^{\omega})$ time?
\end{problem}

\begin{problem}
    \label{prob:2}
    Can \textsf{$k$-APC} be solved in $\tilde{O}((kn)^{\omega})$ time?
\end{problem}

\subsection{Our Contribution}

We resolve all three open problems raised in the previous section. 

First, we present a faster algorithm for \textsf{$k$-APC}, whose time complexity matches the runtime given by previous work for solving an easier version of \textsf{$k$-APVC}.

\begin{restatable}{theorem}{flowthm}
    \label{thm:flow}
    For any positive integer $k$, \textsf{$k$-APC} can be solved in $\tilde{O}((kn)^\omega)$ time.
\end{restatable}
\noindent This is the first algorithm which solves \textsf{$k$-APC} faster than simply solving \textsf{APC} exactly using the $\tilde O(m^\omega)$ time algorithm of \cite{apc}, for all constant $k\ge 3$.

Second, we present an algorithm for \textsf{$k$-APVC}, which is faster than the $\tilde O((kn)^\omega)$ time algorithm from \cite{ap-bounded-mincut} (which only solves a restricted version of \textsf{$k$-APVC}) if $\omega > 2$.

\begin{restatable}{theorem}{vertexflowthm}
    \label{thm:vertex-flow}
    For any positive integer $k$, \textsf{$k$-APVC} can be solved in $\tilde{O}(k^2n^\omega)$ time. 
\end{restatable}

\subsection{Comparison to Previous Results}
\label{subsec:related-work}

\subsubsection*{Conditional Lower Bounds}
The field of fine-grained complexity contains many popular conjectures (which hypothesize lower bounds on the complexity of certain computational tasks) which are used as the basis of conditional hardness results for problems in computer science. 
In this section, we review known hardness results for \textsf{APC} and its variants.
The definitions of the problems and conjectures used in this section are stated in 
\Cref{sec:fg-conjectures}.

Assuming that \textsf{Boolean Matrix Multiplication (BMM)} requires $n^{\omega - o(1)}$ time, it is known that \textsf{$k$-APC} and \textsf{$k$-APVC} require $n^{\omega - o(1)}$ time to solve, even for $k=1$ \cite{transitive-closure-BMM}.
In particular, this hypothesis implies our algorithms for \textsf{$k$-APC} and \textsf{$k$-APVC} are optimal for constant $k$. 

Assuming the \textsf{Strong Exponential Time Hypothesis (SETH)}, previous work shows that \textsf{APC} requires $(mn)^{1 - o(1)}$ time \cite[Theorem 1.8]{seth-flow}, \textsf{APVC} requires $m^{3/2 - o(1)}$ time \cite[Theorem 1.7]{Trabelsi23}, and 
\textsf{$k$-APC} requires $\grp{kn^2}^{1-o(1)}$ time \cite[Theorem 4.3]{seth-flow}.

Let $\omega(1,2,1)$ be the smallest real number\footnote{Known fast matrix multiplication algorithms imply that $\omega(1,2,1) < 3.25669$ \cite[Table 2]{rectangular-mmult}.} such that we can compute the product of an $n\times n^2$ matrix and $n^2\times n$ matrix in $n^{\omega(1,2,1) + o(1)}$ time.
Assuming the \textsf{$4$-Clique Conjecture}, the \textsf{$k$-APVC} problem over directed graphs (and thus the \textsf{$k$-APC} problem as well) requires $\grp{k^2n^{\omega(1,2,1)-2}}^{1-o(1)}$ time \cite{ap-bounded-mincut}.
The \textsf{4-Clique Conjecture} also implies that solving \textsf{APVC} in undirected graphs requires $n^{\omega(1,2,1) - o(1)}$ time 
 \cite{undirected-APVC}.

\subsubsection*{Algorithms for Restricted Graph Classes}
As mentioned previously, no nontrivial algorithms for \textsf{$k$-APC} over general directed graphs were known for $k\ge 3$, prior to our work.
However, faster algorithms were already known for \textsf{$k$-APC} over directed acyclic graphs (DAGs).
In particular, \cite{ap-bounded-mincut} presented two algorithms to solve \textsf{$k$-APC} in DAGs, running in $2^{O(k^2)}mn$ time and $(k\log n)^{4^k + o(k)}n^\omega$ time respectively.

 In comparison, our algorithm from \Cref{thm:flow} solves \textsf{$k$-APC} in \emph{general} directed graphs, is faster than the former algorithm whenever $m\ge n^{\omega - 1}$ or $k\ge \omega(\sqrt{\log n})$ (for example), is always faster than the latter algorithm, and is significantly simpler from a technical perspective than these earlier arguments.
However, these algorithms for \textsf{$k$-APC} on DAGs also return cuts witnessing the connectivity values, while our algorithm does not.

In the special case of undirected graphs, \textsf{APVC} can be solved in $m^{2+o(1)}$ time \cite[Theorem 1.8]{Trabelsi23}, which is faster than the aforementioned $\tilde{O}(m^\omega)$ time algorithm if $\omega > 2$.
Over undirected graphs, \textsf{$k$-APVC} can be solved in $k^3m^{1+o(1)} + n^2\poly(\log n)$ time.
In comparison, our algorithm from \Cref{thm:vertex-flow} can handle \textsf{$k$-APVC} in both undirected \emph{and} directed graphs, and is faster for large enough values of $k$ in dense graphs.

In directed planar graphs with maximum degree $d$, \cite[Theorem 1.5]{apc} proves that \textsf{APC} can be solved in $O\grp{d^{\omega-2}n^{\omega/2+1}}$ time.

\subsubsection*{Additional Related Work}
In \cite{ap-strong-connectivity}, the authors consider a symmetric variant of \textsf{$k$-APC}.
Here, the input is a directed graph $G$ on $n$ vertices and $m$ edges, and the goal is to compute for all pairs of vertices $(s,t)$, the value of $\min(k,\lambda(s,t),\lambda(t,s))$.
This easier problem can be solved in $O(kmn)$ time \cite{ap-strong-connectivity}.

\subsection{Organization}

The rest of this paper is devoted to proving \Cref{thm:flow,thm:vertex-flow}.
In \Cref{sec:prelim} we introduce notation, some useful definitions, and results on matrix computation which will be useful in proving correctness of our algorithms.
In \Cref{sec:technical-overview} we provide an intuitive overview of our algorithms for \textsf{$k$-APC} and \textsf{$k$-APVC}.
In \Cref{sec:flow} we describe a framework of ``flow vectors'' for capturing connectivity values, and in \Cref{sec:k-APC} use this framework to prove \Cref{thm:flow}.
In \Cref{sec:flow-vertex} we present helpful results about vertex-connectivity, and in \Cref{sec:k-APVC} use these results to prove \Cref{thm:vertex-flow}.
We conclude in \Cref{sec:conclusion}, highlighting some interesting open problems suggested by our work.

In \Cref{sec:fg-conjectures}, we include definitions of problems and conjectures mentioned in \Cref{subsec:related-work}.
In \Cref{sec:app-proof}, we discuss how the treatment of \textsf{$k$-APVC} in \cite{ap-bounded-mincut} differs from our own, and present the proof details for one of the results stated in \Cref{sec:flow-vertex}.

\section{Preliminaries}
\label{sec:prelim}

\paragraph*{Graph Assumptions}
Throughout, we let $G$ denote a directed graph on $n$ vertices and $m$ edges.
Without loss of generality, we assume that the underlying undirected graph of $G$ is connected,  i.e., $G$ is weakly connected (since, if not, we could simply run our algorithms separately on each weakly connected component of $G$), so we have $m\ge n-1$.
We  assume $G$ has no self-loops, since these do not affect the connectivity or vertex-connectivity values between distinct vertices.

In \Cref{sec:flow,sec:k-APC} we focus on the \textsf{$k$-APC} problem, and so allow $G$ to have parallel edges between vertices (i.e., $G$ can be a multigraph).
We assume however, without loss of generality, that for any distinct vertices $s$ and $t$, there are at most $k$ edges from $s$ to $t$ (since if there were more than $k$ parallel edges from $s$ to $t$, we could delete some and bring the count of parallel edges down to $k$ without changing the value of $\min(k,\lambda(s,t))$).
In \Cref{sec:flow-vertex,sec:k-APVC} we focus on the \textsf{$k$-APVC} problem, and so assume that $G$ is a simple graph with no parallel edges, since parallel edges from $u$ to $v$ cannot affect the value of a vertex connectivity $\nu(s,t)$, unless $u=s$ and $v=t$, in which case the value of $\nu(s,t)$ is simply increased by the number of additional parallel edges from $s$ to $t$. 

\paragraph*{Graph Terminology and Notation}
Given an edge $e$ from $u$ to $v$ in $G$, we write $e = (u,v)$.
We call $u$ the tail of $e$ and $v$  the head of $e$.
Vertices which are tails of edges entering a vertex $v$ are called \emph{in-neighbors} of $v$.
Similarly, vertices which are heads of edges exiting $v$ are called \emph{out-neighbors} of $v$.
Given a vertex $u$ in $G$, we let $\inedges{u}$ denote the set of edges entering $u$, and $\outedges{u}$ denote the set of edges exiting $u$.
Similarly, $\innodes{u}$ denotes the set of in-neighbors of $u$, and $\outnodes{u}$ denotes the set of out-neighbors of $u$. 
Furthermore, we define $\innodescl{u} = \innodes{u}\cup\set{u}$ and $\outnodescl{u} = \outnodes{u}\cup \set{u}$.
Finally, let $\indeg(u) = |\inedges{u}|$ and $\outdeg(u) = |\outedges{u}|$ denote the indegree and outdegree of $u$ respectively.

Given vertices $s$ and $t$, an $(s,t)$-cut is a set $C$ of edges, such that deleting the edges in $C$ produces a graph with no $s$ to $t$ path.
By Menger's theorem, the size of a minimum $(s,t)$-cut is equal to the connectivity $\lambda(s,t)$ from $s$ to $t$. 
Similarly, an $(s,t)$-vertex cut is a set of $C'$ of vertices with $s,t\not\in C'$, such that deleting $C'$ produces a graph with no $s$ to $t$ path.
Clearly, a vertex cut exists if and only if $(s,t)$ is not an edge.
When $(s,t)$ is not an edge, Menger's theorem implies that the size of a minimum $(s,t)$-vertex cut is equal to the vertex connectivity $\nu(s,t)$ from $s$ to $t$.

\paragraph*{Matrix Notation}
Let $A$ be a matrix.
For indices $i$ and $j$, we let $A[i,j]$ denote the $(i,j)$ entry of $A$.
More generally, if $S$ is a set of row indices and $T$ a set of column indices, we let $A[S,T]$ denote the submatrix of $A$ restricted to rows from $S$ and columns from $T$.
Similarly, $A[S,\ast]$ denotes $A$ restricted to rows from $S$, and $A[\ast, T]$ denotes $A$ restricted to columns from $T$. 
We let $A^\top$ denote the transpose of $A$.
If $A$ is a square matrix, then we let $\adj(A)$ denote the adjugate of $A$.
If $A$ is invertible, we let $A^{-1}$ denote its inverse.
If a theorem, lemma, or proposition statement refers to $A^{-1}$, it is generally asserting that $A^{-1}$ exists (or if $A$ is a random matrix, asserting that $A^{-1}$ exists with some probability) as part of the statement.
We let $I$ denote the identity matrix (the dimensions of this matrix will always be clear from context).
Given a vector $\vec{v}$, for any index $i$ we let $\vec{v}[i]$ denote the $i^{\text{th}}$ entry in $\vec{v}$.
We let $\vec{0}$ denote the zero vector (the dimensions of this vector will always be clear from context).
Given a positive integer $k$, we let $[k] = \set{1, \dots, k}$ denote the set of the first $k$ positive integers.

\paragraph*{Matrix and Polynomial Computation}

Given a prime $p$, we let $\mathbb{F}_p$ denote the finite field on $p$ elements.
Arithmetic operations over elements of $\mathbb{F}_p$ can be performed in $\tilde{O}(\log p)$ time. 

We now recall some well-known results about computation with matrices and polynomials, which will be useful for our algorithms.

\begin{proposition}
\label{lem:inverse}
Let $A$ be an $a\times b$ matrix, and $B$ be a $b\times a$ matrix.
If $(I-BA)$ is invertible, then the matrix $(I-AB)$ is also invertible, with inverse
    \[(I-AB)^{-1} = I + A(I-BA)^{-1}B.\]
\end{proposition}
\begin{proof}
    It suffices to verify that the product of $(I-AB)$ with the right hand side of the above equation yields the identity matrix. 
    Indeed, we have 
    \begin{align*}
     & (I-AB)\grp{I + A(I-BA)^{-1}B} \\= \ & I + A(I-BA)^{-1}B - AB - ABA(I-BA)^{-1}B 
    \\= \ &   I + A(I-BA)^{-1}B - AB - A\big (I - (I-BA)\big )(I-BA)^{-1}B 
    \\= \ &  I + A(I-BA)^{-1}B - AB - A(I-BA)^{-1}B + AB,
    \end{align*}
    which simplifies to $I$, as desired.
\end{proof}

\begin{proposition}
    \label{lem:compute-inverse}
    Let $A$ be an $a\times a$ matrix over $\mathbb{F}_p$.
    We can compute the inverse $A^{-1}$ (if it exists) in $O(a^{\omega})$ field operations.
\end{proposition}

\begin{proposition}[{\cite[Theorem 1.1]{fast-matrix-rank}}]
    \label{lem:rank}
    Let $A$ be an $a\times b$ matrix over $\mathbb{F}_p$.
    Then for any positive integer $k$, we can compute $\min(k,\rank A)$ in $O(ab+k^\omega)$ field operations.
\end{proposition}

\begin{proposition}[Schwartz-Zippel Lemma {\cite[Theorem 7.2]{randomized-book}}]
    \label{schwartz-zippel}
    Let $f \in \mathbb{F}_p[x_1, \dots, x_r]$ be a degree $d$, nonzero polynomial.
    Let $\vec{a}$ be a uniform random point in $\mathbb{F}_p^r$.
    Then $f(\vec{a})$ is nonzero with probability at least $1-d/p$.
\end{proposition}

\section{Proof Overview}
\label{sec:technical-overview}

\subsection{Flow Vector Encodings}
\label{subsec:vector-encoding}

Previous algorithms for \textsf{APC} \cite{apc} and its variants work in two steps:

\begin{description}
    \item[Step 1: Encode]\hfill
    
        In this step, we prepare a matrix $M$ which implicitly encodes the connectivity information of the input graph.
    
    \item[Step 2: Decode]\hfill

        In this step, we iterate over all pairs $(s,t)$ of vertices in the graph, and for each pair run a small computation on a submatrix of $M$ to compute the desired connectivity value.
\end{description}

\noindent The construction in the \textbf{encode} step is based off the framework of \emph{flow vectors}, introduced in \cite{apc} as a generalization of classical techniques from network-coding. 
We give a high-level overview of how this method has been previously applied in the \textsf{APC} problem.\footnote{For the \textsf{APVC} problem\ we employ a different, but analogous, framework  described in \cref{subsec:apvc}. }

Given the input graph $G$, we fix a source vertex $s$. 
Let $d = \outdeg(s)$, and let $\mathbb{F}$ be some ground field.\footnote{In our applications, we will pick $\mathbb{F}$ to be a finite field of size $\poly(m)$.}
Our end goal is to assign to each edge $e$ in the graph a special vector $\vec{e}\in\mathbb{F}^d$ which we call a \emph{flow vector}.

First, for each edge $e \in \outedges{s}$, we introduce a $d$-dimensional vector $\vec{v_e}$.
These vectors intuitively correspond to some starting flow that is pumping out of $s$. 
It is important that these vectors are linearly independent (and previous applications have always picked these vectors to be distinct $d$-dimensional unit vectors). 
We then push this flow through the rest of the graph, 
by having each edge get assigned a vector which is a random linear combination of the flow vectors assigned to the edges entering its tail.
That is, given an edge $e = (u,v)$ with $u\neq s$, the final flow vector $\vec{e}$ will be a random linear combination of the flow vectors for the edges entering $u$. 
If instead the edge $e= (s,v)$ is in $\outedges{s}$, the final flow vector $\vec{e}$ will be a random linear combination of the flow vectors for the edges entering $s$, added to the initial flow $\vec{v}_e$.

The point of this random linear combination is to (with high probability) preserve linear independence.
In this setup, for any vertex $v$ and integer $\ell$, if some subset of $\ell$ flow vectors assigned to edges in $\inedges{v}$ is independent, then we expect that every subset of at most $\ell$ flow vectors assigned to edges in $\outedges{v}$ is also independent.
This sort of behavior turns out to generalize to preserving linear independence of flow vectors across cuts, which implies that (with high probability) for any vertex $t$, the rank of the flow vectors assigned to edges in $\inedges{t}$ equals $\lambda(s,t)$.

Intuitively, this is because the flow vectors assigned to edges in $\inedges{t}$ will be a linear combination of the $\lambda(s,t)$ flow vectors assigned to edges in a minimum $(s,t)$-cut, and the flow vectors assigned to edges in this cut should be independent.

Collecting all the flow vectors as column vectors in a matrix allows us to produce a single matrix $M_s$, such that computing the rank of $M_s[\ast, \inedges{t}]$ yields the desired  connectivity value $\lambda(s,t)$ (computing these ranks constitutes the \textbf{decode} step mentioned previously).
Previous work \cite{apc,ap-bounded-mincut} set the initial pumped $\vec{v}_e$ to be distinct unit vectors.
It turns out that for this choice of starting vectors, it is possible to construct a single matrix $M$ (independent of a fixed choice of $s$), such that rank queries to submatrices of $M$ correspond to the answers we wish to output in the \textsf{APC} problem and its variants.

In \Cref{subsec:apc} we describe how we employ the flow vector framework to prove \Cref{thm:flow}.
Then in \Cref{subsec:apvc}, we describe how we modify these methods to prove \Cref{thm:vertex-flow}.

\subsection{All-Pairs Connectivity}
\label{subsec:apc}

Our starting point is the $\tilde{O}(m^\omega)$ time algorithm for \textsf{APC} presented in \cite{apc}, which uses the flow vector encoding scheme outlined in \Cref{subsec:vector-encoding}.

Let $K$ be an $m\times m$ matrix, whose rows and columns are indexed by edges in the input graph.
For each pair $(e,f)$ of edges, if the head of $e$ coincides with the tail of $f$, we set $K[e,f]$ to be a uniform random field element in $\mathbb{F}$.
Otherwise, $K[e,f] = 0$.
These field elements correspond precisely to the coefficients used in the random linear combinations of the flow vector framework.
Define the matrix
    \begin{equation}
    \label{eq:old-apc-def}
    M = (I - K)^{-1}.
    \end{equation}
Then \cite{apc} proves that with high probability, for any pair $(s,t)$ of vertices, we have 
    \begin{equation}
    \label{eq:old-apc-rank}
    \rank M[\outedges{s},\inedges{t}] = \lambda(s,t).
    \end{equation}
With this setup, the algorithm for \textsf{APC} is simple: first compute $M$ (the \textbf{encode} step), and then for each pair of vertices $(s,t)$, return the value of $\rank M[\outedges{s},\inedges{t}]$ as the connectivity from $s$ to $t$ (the \textbf{decode} step). 

By \Cref{eq:old-apc-def}, we can complete the \textbf{encode} step in  $\tilde{O}(m^\omega)$ time, simply by inverting an $m\times m$ matrix with entries from $\mathbb{F}$.
It turns out we can also complete the \textbf{decode} step in the same time bound. 
So this gives an $\tilde{O}(m^\omega)$ time algorithm for \textsf{APC}.

Suppose now we want to solve the \textsf{$k$-APC} problem.
A simple trick (observed in the proof of \cite[Theorem 5.2]{ap-bounded-mincut} for example) in this setting can allow us to speed up the runtime of the \textbf{decode} step.
However, it is not at all obvious how to speed up the \textbf{encode} step.
To implement the flow vector scheme of \Cref{subsec:vector-encoding} as written, it seems almost inherent that one needs to  invert an $m\times m$ matrix.
Indeed, an inability to overcome this bottleneck is stated explicitly as part of the motivation in \cite{ap-bounded-mincut} for focusing on the \textsf{$k$-APVC} problem instead. 

\subsubsection*{Our Improvement}

The main idea behind our new algorithm for \textsf{$k$-APC} is to work with a low-rank version of the matrix $K$ used in \cref{eq:old-apc-def} for the \textbf{encode} step.

Specifically, we construct certain random sparse matrices $L$ and $R$ with dimensions $m\times kn$ and $kn\times m$ respectively.
We then set $K = LR$, and argue that with high probability, the matrix $M$ defined in \cref{eq:old-apc-def} for this choice of $K$ satisfies 
    \begin{equation}
    \label{eq:low-rank-intuition}
    \rank M[\outedges{s},\inedges{t}] = \min(k,\lambda(s,t)).
    \end{equation}
This equation is just a $k$-bounded version of \cref{eq:old-apc-rank}.
By \Cref{lem:inverse}, we have 
    \[M = (I-K)^{-1} = (I-LR)^{-1} = I + L(I-RL)^{-1}R.\]
Note that $(I-RL)$ is a $kn\times kn$ matrix.
So, to compute $M$ (and thus complete the \textbf{encode} step) we no longer need to invert an $m\times m$ matrix!
Instead we just need to invert a matrix of size $kn\times kn$.
This is essentially where the $\tilde{O}\grp{(kn)^{\omega}}$ runtime in \Cref{thm:flow} comes from.

Conceptually, this argument corresponds to assigning flow vectors through the graph by replacing random linear combinations with random ``low-rank combinations.'' 
That is, for an edge $e \in \outedges{u}$ exiting a vertex $u$, we define the flow vector at $e$ to be 
    \[\vec{e} = \sum_{i=1}^k\grp{\sum_{f\in \inedges{u}} L_i[f,u]\vec{f}}\cdot R_i[u,e],\]
where the inner summation is over all edges $f$ entering $u$, $\vec{f}$ denotes the flow vector assigned to edge $f$, and the $L_i[f,u]$ and $R_i[u,e]$ terms correspond to random field elements uniquely determined by the index $i$ and some $(\text{edge}, \text{vertex})$ pair. 

Here, unlike in the method described in \Cref{subsec:vector-encoding}, the coefficient in front of $\vec{f}$ in its contribution to $\vec{e}$ is not uniquely determined by the pair of edges $f$ and $e$.
Rather, if edge $f$ enters node $u$, then it has the same set of ``weights'' $L_i[f,u]$ it contributes to every flow vector exiting $u$. 
However, since we use $k$ distinct weights, this restricted rule for propagating flow vectors still suffices to compute $\min(k,\lambda(s,t))$.

A good way to think about the effect of this alternate approach is that now for any vertex $v$ and any integer $\ell\le k$, if some subset of $\ell$ flow vectors assigned to edges in $\inedges{v}$ is independent, then we expect that every subset of at most $\ell$ flow vectors assigned to edges in $\outedges{v}$ is also independent.
In the previous framework, this result held even for $\ell > k$.
By relaxing the method used to determine flow vectors, we achieve a weaker condition, but this is still enough to solve \textsf{$k$-APC}.

This modification makes the \textbf{encode} step more complicated (it now consists of two parts: one where we invert a matrix, and one where we multiply that inverse with other matrices), but speeds it up overall.
To speed up the \textbf{decode} step, we use a variant of an observation from the proof of \cite[Theorem 5.2]{ap-bounded-mincut} to argue that we can assume every vertex in our graph has indegree and outdegree $k$. 
By \Cref{lem:rank} and \cref{eq:low-rank-intuition}, this means we can compute $\min(k,\lambda(s,t))$ for all pairs $(s,t)$ of vertices in $\tilde{O}(k^{\omega}n^2)$ time.
So the bottleneck in our algorithm comes from the \textbf{encode} step, which yields the $\tilde{O}\grp{(kn)^\omega}$ runtime.

\subsection{All-Pairs Vertex Connectivity}
\label{subsec:apvc}

Our starting point is the $\tilde{O}\grp{(kn)^{\omega}}$ time algorithm in \cite{ap-bounded-mincut}, which computes $\min(k,\nu(s,t))$ for all pairs of vertices $(s,t)$ which are not edges.
That algorithm is based off a variant of the flow vector encoding scheme outlined \Cref{subsec:vector-encoding}.
Rather than assign vectors to edges, we instead assign flow vectors to vertices (intuitively this is fine because we are working with vertex connectivities in the \textsf{$k$-APVC} problem).
The rest of the construction is similar: we imagine pumping some initial vectors to $s$ and its out-neighbors, and then we propagate the flow through the graph so that at the end, for any vertex $v$, the flow vector assigned to $v$ is a random linear combination of flow vectors assigned to in-neighbors of $v$.\footnote{Actually, this behavior only holds for vertices $v\not\in\outnodescl{s}.$ The rule for flow vectors assigned to vertices in $\outnodescl{s}$ is a little more complicated, and depends on the values of the initial pumped vectors.}  

Let $K$ be an $n\times n$ matrix, whose rows and columns are indexed by vertices in the input graph.
For each pair $(u,v)$ of vertices, if there is an edge from $u$ to $v$, we set $K[u,v]$ to be a uniform random element in $\mathbb{F}$.
Otherwise, $K[u,v] = 0$.
These entries correspond precisely to coefficients used in the random linear combinations of the flow vector framework.

Now define the matrix
    \begin{equation}
    \label{apvc-M}
    M = (I-K)^{-1}.
    \end{equation}
Then we argue that for any pair $(s,t)$ of vertices, we have 
    \begin{equation}
    \label{apvc-ranks}
    \rank M[\outnodescl{s},\innodescl{t}] = \begin{cases}\nu(s,t) + 1 & \text{if }(s,t)\text{ is an edge} \\ \nu(s,t) & \text{otherwise}. \end{cases}
    \end{equation}
Previously, \cite[Proof of Lemma 5.1]{ap-bounded-mincut} sketched a different argument, which shows that $\rank M[\outnodes{s},\innodes{t}] = \nu(s,t)$ when $(s,t)$ is not an edge.
As we discuss in \Cref{correctingtheliterature}, this  claim does not necessarily hold when $(s,t)$ is an edge.

We use \Cref{apvc-ranks} to solve \textsf{$k$-APVC}.
For the \textbf{encode} step, we compute $M$.
By \cref{apvc-M}, we can do this by inverting an $n\times n$ matrix, which takes $\tilde{O}(n^\omega)$  time.
For the \textbf{decode} step, by \cref{apvc-ranks} and \Cref{lem:rank}, we can compute $\min(k,\nu(s,t))$ for all pairs $(s,t)$ of vertices in  asymptotically
    \[\sum_{s, t} \grp{\outdeg(s)\indeg(t) + k^\omega} = m^2 + k^\omega n^2\]
    time,
where the sum is over all vertices $s$ and $t$ in the graph.
The runtime bound we get here for the \textbf{decode} step is far too high -- naively computing the ranks of submatrices is too slow if the graph has many high-degree vertices.

To avoid this slowdown, \cite{ap-bounded-mincut} employs a simple trick to reduce degrees in the graph: we can add layers of $k$ new nodes to block off the ingoing and outgoing edges from each vertex in the original graph.
That is, for each vertex $s$ in $G$, we add a set $S$ of $k$ new nodes, 
 replace the edges in $\outedges{s}$ with edges from  $s$ to all the nodes in $S$, and add edges from every node in $S$ to every vertex originally in $\outnodes{s}$.
Similarly, for each vertex $t$ in $G$, we add a set $T$ of $k$ new nodes, 
 replace the edges in $\inedges{t}$ with edges from all the nodes in $T$ to $t$, and add edges from every vertex originally  in $\innodes{t}$ to every node in $T$.

It is easy to check that this transformation preserves the value of $\min(k,\nu(s,t))$ for all pairs $(s,t)$ of vertices in the original graph where $(s,t)$ is not an edge.
Moreover, all vertices in the original graph have indegree and outdegree exactly $k$ in the new graph.
Consequently, the \textbf{decode} step can now be implemented to run in $\tilde{O}(k^\omega n^2)$ time.
Unfortunately, this construction increases the number of vertices in the graph from $n$ to $(2k+1)n$.
As a consequence, in the \textbf{encode} step, the matrix $K$ we work with is no longer $n\times n$, but instead is of size $(2k+1)n\times (2k+1)n$.
Now  inverting $I-K$ to compute $M$ requires $\tilde{O}\grp{(kn)^\omega}$ time, which is why \cite{ap-bounded-mincut} obtains this runtime for their algorithm.

\subsubsection*{Our Improvement}

Intuitively, the modification used by \cite{ap-bounded-mincut} to reduce degrees in the graph feels very inefficient.
This transformation makes the graph larger in order to ``lose information'' about connectivity values greater than $k$.
Rather than modify the graph in this way, can we modify the flow vector scheme itself to speed up the \textbf{decode} step? 
Our algorithm does this, essentially modifying the matrix of flow vectors to simulate the effect of the previously described transformation, without ever explicitly adding  new nodes to the graph.

Instead of working directly with the matrix $M$ from \Cref{apvc-M}, for each pair $(s,t)$ of vertices we define a $(k+1)\times (k+1)$ matrix 
    \[M_{s,t} = B_s \grp{M[\outnodescl{s},\innodescl{t}]} C_t\]
which is obtained from multiplying a submatrix of $M$ on the left and right by small random matrices $B_s$ and $C_t$, with $k+1$ rows and columns respectively.
Since $B_s$ has $k+1$ rows and $C_t$ has $k+1$ columns, we can argue that with high probability, \Cref{apvc-ranks} implies that
    \[\rank M_{s,t} = \begin{cases} \min(k+1,\nu(s,t)+1) & \text{if }(s,t)\text{ is an edge} \\
    \min(k+1,\nu(s,t)) & \text{otherwise}.\end{cases}\]
So we can compute $\min(k,\nu(s,t))$ from the value of $\rank M_{s,t}$.
This idea is similar to the preconditioning method used in algorithms for computing matrix rank efficiently (see \cite{fast-matrix-rank} and the references therein).
Conceptually, we can view this approach as a modification of the flow vector framework.
Let $d = \outdeg(s)$.
As noted in \Cref{subsec:vector-encoding}, previous work 

\begin{enumerate}
\item starts by pumping out distinct $d$-dimensional unit vectors to nodes in $\outnodes{s}$, and then
\item computes the rank of all flow vectors of vertices in $\innodes{t}$.
\end{enumerate}

\noindent In our work, we instead 

\begin{enumerate}
    \item start by pumping out $(d+1)$ random $(k+1)$-dimensional vectors to nodes in $\outnodescl{s}$, and then
    \item
     compute the rank of $(k+1)$ random linear combinations of flow vectors for vertices in $\innodescl{t}$.
\end{enumerate}

\noindent This alternate approach suffices for solving the \textsf{$k$-APVC} problem, while avoiding the slow $\tilde{O}((kn)^\omega)$ \textbf{encode} step of previous work.

So, in the \textbf{decode} step of our algorithm, we compute $\min(k,\nu(s,t))$ for each pair $(s,t)$ of vertices by computing the rank of the $(k+1)\times (k+1)$ matrix $M_{s,t}$, in $\tilde{O}(k^\omega n^2)$ time overall.

Our \textbf{encode} step is more complicated than previous work, because not only do we need to compute the inverse $(I-K)^{-1}$, we also have to construct the $M_{s,t}$ matrices.
Naively computing each $M_{s,t}$ matrix separately is too slow, so we end up using an indirect approach to compute all entries of the $M_{s,t}$ matrices \emph{simultaneously}, with just $O(k^2)$ multiplications of $n\times n$ matrices.
This takes $\tilde{O}(k^2n^\omega)$ time, which is the bottleneck for our algorithm.


\section{Flow Vector Encoding}
\label{sec:flow}

The arguments in this section are similar to the arguments from \cite[Section 2]{apc}, but involve more complicated proofs because we work with low-rank random matrices as opposed to generic random matrices.

Fix a source vertex $s$ in the input graph $G$. 
Let $d = \outdeg(s)$ denote the number of edges leaving $s$.
Let $e_1, \dots, e_d \in \outedges{s}$ be the outgoing edges from $s$.

Take a prime $p = \Theta(m^5)$.
Let $\vec{u}_1, \dots, \vec{u}_d$ be distinct unit vectors in $\mathbb{F}_p^d$.

Eventually, we will assign each edge $e$ in $G$ a  vector $\vec{e}\in\mathbb{F}_p^d$, which we call a \emph{flow vector}.
These flow vectors will be determined by a certain system of vector equations.
To describe these equations, we first introduce some symbolic matrices.

For each index $i \in [k]$, we define an $m\times n$ matrix $X_i$, whose rows are indexed by edges of $G$ and columns are indexed by vertices of $G$, such that for each edge $e = (u,v)$, entry $X_i[e,v] = x_{i,ev}$ is an indeterminate.
All entries in $X_i$ not of this type are zero.

Similarly, we define $n\times m$ matrices $Y_i$, with rows indexed by vertices of $G$ and columns indexed by edges of $G$, such that for every edge $f = (u,v)$, the entry $Y_i[u,f] = y_{i,uf}$ is an indeterminate.
All entries in $Y_i$ not of this type are zero.

Let $X$ be the $m\times kn$ matrix formed by horizontally concatenating the $X_i$ matrices.
Similarly, let $Y$ be the $kn\times m$ matrix formed by vertically concatenating the $Y_i$ matrices.
Then we define the matrix
    \begin{equation}
    \label{eq:symbolic-transfer}
    Z = XY = X_1Y_1 + \dots + X_kY_k.
    \end{equation}
By construction, $Z$ is an $m\times m$ matrix, with rows and columns indexed by edges of $G$, such that for any edges $e = (u,v)$ and $f = (v,w)$, we have 
    \begin{equation}
    \label{eq:symbolic-entries}
    Z[e,f] = \sum_{i=1}^k x_{i,ev}y_{i,vf}
    \end{equation}
and all other entries of $Z$ are set to zero.

Consider the following procedure.
We assign independent, uniform random values from $\mathbb{F}_p$ to each variable $x_{i,ev}$ and $y_{i,uf}$.
Let $L_i, L, R_i, R$, and $K$ be the matrices over $\mathbb{F}_p$ resulting from this assignment to $X_i, X, Y_i, Y$, and $Z$ respectively.
In particular, we have
\begin{equation}
\label{claim:low-rank}
    K = LR.
\end{equation}
Now, to each edge $e$, we assign a flow vector $\vec{e} \in  \mathbb{F}_p^d $, satisfying the following equalities:

\begin{enumerate}
    \item 
    Recall that $e_1, \dots, e_d$ are all the edges exiting $s$, and $\vec{u}_1,\dots,\vec{u}_d$ are distinct unit vectors in $\mathbb{F}_p^d$. 
    For each edge $e_i \in\outedges{s}$, we require its flow vector satisfy
        \begin{equation}
        \label{eq:flow-sout}
        \vec{e}_i = \left(\sum_{
        f\in \inedges{s}
        } \vec{f}\cdot K[f,e_i]\right) +  \vec{u}_i.
        \end{equation}

    \item
    For each edge $e = (u,v)$ with $u\neq s$, we require its flow vector satisfy
            \begin{equation}
            \label{eq:flow-general}
            \vec{e} = \sum_{f\in \inedges{u}} \vec{f}\cdot K[f,e].
            \end{equation}
\end{enumerate}
A priori it is not obvious that flow vectors satisfying the above two conditions exist, but we show below that they do (with high probability).
Let $H_s$ be the $d\times m$ matrix whose columns are indexed by edges in $G$, 
such that the column associated with $e_i$ is $\vec{u}_i$ for each index $i$, and the rest of the columns are zero vectors. 
Let $F$ be the $d\times m$ matrix, with columns indexed by edges in $G$, whose columns $F[\ast,e] = \vec{e}$ are flow vectors for the corresponding edges.
Then \cref{eq:flow-sout,eq:flow-general} are encapsulated by the simple matrix equation
    \begin{equation}
    \label{eq:flow-matrix-equation}
    F = FK + H_s.
    \end{equation}
The following lemma shows we can solve for $F$ in the above equation, with high probability.

\begin{lemma}
    \label{lem:invertible}
    We have $\det(I-K)\neq 0$, with probability at least $1-1/m^3$.
\end{lemma}
\begin{proof}
    Since the input graph has no self-loops, by \cref{eq:symbolic-entries} and the discussion immediately following it, we know that the diagonal entries of the $m\times m$ matrix $Z$ are zero.
    By \cref{eq:symbolic-entries}, each entry of $Z$ is a polynomial of degree at most two, with constant term set to zero. Hence, $\det(I-Z)$ is a polynomial over $\mathbb{F}_p$ with degree at most $2m$, and constant term equal to $1$.
    In particular, this polynomial is nonzero.
    Then by the Schwartz-Zippel Lemma (\Cref{schwartz-zippel}), $\det(I-K)$ is nonzero with probability at least 
    \[1-2m/p \ge 1-1/m^3\]
    by setting $p\ge 2m^4$.
\end{proof}

    Suppose from now on that $\det(I-K)\neq 0$ (by \cref{lem:invertible}, this occurs with high probability).
    Then with this assumption, we can solve for $F$ in \cref{eq:flow-matrix-equation} to get

        \begin{equation}
        \label{eq:flow-closed-form}
        F = H_s(I-K)^{-1} = \frac{H_s\grp{\text{adj}(I-K)}}{\det(I-K)}.
        \end{equation}
    This equation will allow us to relate ranks of collections of flow vectors to connectivity values in the input graph.

    \begin{lemma}
    \label{lem:conn-upper-bound}
    For any vertex $t$ in $G$, with probability at least $1-2/m^3$, we have
        \[\rank F[\ast, \inedges{t}] \le \lambda(s,t).\]
    \end{lemma}
    \begin{proof}
        Abbreviate $\lambda = \lambda(s,t)$.
            Conceptually, this proof works by arguing that the flow vectors assigned to all edges entering $t$ are linear combinations of the flow vectors assigned to edges in a minimum $(s,t)$-cut of $G$.
        
        Let $C$ be a minimum $(s,t)$-cut.
        By Menger's theorem, $|C|=\lambda$.

        Let $S$ be the set of nodes reachable from $s$ without using an edge in $C$, and let $T$ be the set of nodes which can reach $t$ without using an edge in $C$. 
        By definition of an $(s,t)$-cut, $S$ and $T$ partition the vertices in $G$.
        
        Let $E'$ be the set of edges $e = (u,v)$ with $v\in T$.
        
        Set $K' = K[E', E']$ and $F' = F[\ast, E'].$
        Finally, let $H'$ be a matrix whose columns are indexed by edges in $E'$, such that the column associated with an edge $e\in C$ is $\vec{e}$, and all other columns are equal to $\vec{0}$.

        Then by \cref{eq:flow-sout,eq:flow-general}, we have 
            \[F' = F'K' + H'.\]
        Indeed, for any edge $e = (u,v)\in E'$, if $u\in S$ then $e\in C$ so $H'[\ast, e] = \vec{e}$, and there can be no edge $f\in E'$ entering $u$, so $(F'K')[\ast, e] = \vec{0}$.
        If instead $u\in T$, then $H'[\ast, e] = \vec{0}$, but every edge $f$ entering $u$ is in $E'$, so by \cref{eq:flow-general}, we have $(F'K')[\ast,e] = F'[\ast,e]$ as desired.

        Using similar reasoning to the proof of \cref{lem:invertible}, we have $\det (I-K')\neq 0$ with probability at least $1-1/m^3$.
        If this event occurs, we can solve for $F'$ in the previous equation to get 
            \[F' = H'(I-K')^{-1}.\]
        Since $H'$ has at most $\lambda$ nonzero columns, $\rank H\le \lambda$.
        So by the above equation, $\rank F'\le \lambda$.
        By definition, $\inedges{t}\sub E'$.
        Thus $F[\ast,\inedges{t}]$ is a submatrix of $F'$.
        Combining this with the previous results, we see that $\rank F[\ast,\inedges{t}]\le \lambda$, as desired.
        The claimed probability bound follows by a union bound over the events that $I-K$ and $I-K'$ are both invertible.
    \end{proof}

    \begin{lemma}
    \label{lem:conn-lower-bound}
    For any vertex $t$ in $G$, with probability at least $1-2/m^3$, we have 
    \[\rank F[\ast, \inedges{t}] \ge \min(k, \lambda(s,t)).\]
    \end{lemma}
    \begin{proof}

    Abbreviate $\lambda = \min(k,\lambda(s,t))$.
        Intuitively, our proof argues that the presence of edge-disjoint paths from $s$ to $t$ leads to certain edges in $\inedges{t}$ being assigned linearly independent flow vectors (with high probability), which then implies the desired lower bound.
        
    By Menger's theorem, $G$ contains $\lambda$ edge-disjoint paths $P_1, \dots, P_\lambda$ from $s$ to $t$.

    Consider the following assignment to the variables of the symbolic matrices $X_i$ and $Y_i$.
    For each index $i\le \lambda$ and edge $e = (u,v)$, we set variable $x_{i,ev} = 1$ if $e$ is an edge in $P_i$.
    Similarly, for each index $i\le \lambda $ and edge $f = (u, v)$, we set variable $y_{i, uf} = 1$ if $f$ is an edge in $P_i$.
    All other variables are set to zero. 
    In particular, if $i > \lambda$, then $X_i$ and $Y_i$ have all their entries set to zero. 
    With respect to this assignment, the matrix
    $X_iY_i$
    (whose rows and columns are indexed by edges in the graph) has the property that $(X_iY_i)[e,f] = 1$ if $f$ is the edge following $e$  on path $P_i$, and all other entries are set to zero.

    Then by \cref{eq:symbolic-transfer}, we see that under this assignment, $Z[e,f] = 1$ if $e$ and $f$ are consecutive edges in some path $P_i$, and all other entries of $Z$ are set to zero.
    For this particular assignment, because the $P_i$ are edge-disjoint paths, \cref{eq:flow-sout,eq:flow-general} imply that the last edge of each path $P_i$ is assigned a distinct $d$-dimensional unit vector.
    These vectors are independent, so, $\rank F[\ast, \inedges{t}] = \lambda$ in this case.
    
    With respect to this assignment, this means that $F[\ast, \inedges{t}]$ contains a $\lambda \times \lambda$ full-rank submatrix.
    Let $F'$ be a submatrix of $F[\ast,\inedges{t}]$ with this property.
        Since $F'$ has full rank, we have $\det F'\neq 0$ for the assignment described above.

    Now, before assigning values to variables, each entry of $\text{adj}(I-Z)$ is a polynomial of degree at most $2m$.
    So by \cref{eq:flow-closed-form}, $\det F'$ is equal to some polynomial $P$ of degree at most $2\lambda m$, divided by $(\det(I-Z))^{\lambda}$.
    We know $P$ is a nonzero polynomial, because we saw above that $\det F'$ is nonzero for some assignment of values to the variables (and if $P$ were the zero polynomial, then $\det F'$ would evaluate to zero under every assignment).

    By \cref{lem:invertible}, with probability at least $1-1/m^3$, a random evaluation to the variables will have $\det(I-Z)$ evaluate to a nonzero value. 
    Assuming this event occurs, by Schwartz-Zippel Lemma (\cref{schwartz-zippel}), a random evaluation to the variables in $Z$ will have $\det F'\neq 0$ with probability at least 
        $1 - (2\lambda m)/p \ge 1-1/m^3$
    by setting $p\ge 2m^5.$

    So by union bound, a particular $\lambda \times \lambda$ submatrix of $F[\ast, \inedges{t}]$ will be full rank with probability at least $1-2/m^3$.
    This proves the desired result. 
    \end{proof}

    \begin{lemma}
    \label{lem:conn-equality}
    Fix vertices $s$ and $t$.
    Define $\lambda = \rank~(I-K)^{-1}[\outedges{s},\inedges{t}].$
    With probability at least $1 - 4/m^3$, we have 
        $\min(k,\lambda) = \min(k,\lambda(s,t)).$
    \end{lemma}
    \begin{proof}
        The definition of $H_s$ together with 
        \cref{eq:flow-closed-form} implies that 
            \begin{equation}
    \label{eq:flow-submatrix}
    F[\ast, \inedges{t}] = (I-K)^{-1}[\outedges{s},\inedges{t}].
    \end{equation}
    By union bound over \cref{lem:conn-lower-bound,lem:conn-upper-bound}, with probability at least $1 - 4/m^3$ the inequalities
        \[\lambda = \rank~(I-K)^{-1}[\outedges{s},\inedges{t}] = \rank F[\ast,\inedges{t}] \le \lambda(s,t)\]
    and 
        \[\lambda = \rank~(I-K)^{-1}[\outedges{s},\inedges{t}] = \rank F[\ast,\inedges{t}] \ge \min(k,\lambda(s,t)) \]
        simultaneously hold.
        The desired result follows.
    \end{proof}



\section{Connectivity Algorithm} 
\label{sec:k-APC}

In this section, we present our algorithm for \textsf{$k$-APC}.

\paragraph*{Graph Transformation}

We begin by modifying the input graph $G$ as follows.
For every vertex $v$ in $G$, we introduce two new nodes $v_{\t{out}}$ and $v_{\t{in}}$.
We replace each edge $(u,v)$ originally in $G$ is by the edge $(u_{\t{out}},v_{\t{in}})$.
We add $k$ parallel edges from $v$ to $v_{\t{out}}$, and $k$ parallel edges from $v_{\t{in}}$ to $v$, for all $u$ and $v$.
We call vertices present in the graph before modification the \emph{original vertices}.

Suppose $G$ originally had $n$ nodes and $m$ edges.
Then the modified graph has $\nnew = 3n$ nodes and $\mnew = m+2kn$ edges.
For any original vertices $s$ and $t$, edge-disjoint paths from $s$ to $t$ in the new graph correspond to edge disjoint paths from $s$ to $t$ in the original graph.
Moreover, for any integer $\ell\le k$, if the original graph contained $\ell$ edge-disjoint paths from $s$ to $t$, then the new graph contains $\ell$ edge-disjoint paths from $s$ to $t$ as well.

Thus, for any original vertices $s$ and $t$, the value of $\min(k, \lambda(s,t))$ remains the same in the old graph and the new graph. 
So, it suffices to solve \textsf{$k$-APC} on the new graph.
In this new graph, the indegrees and outdegrees of every original vertex are equal to $k$.
Moreover, sets $\outedges{s}$ and $\inedges{t}$ are pairwise disjoint, over all original vertices $s$ and $t$.

\paragraph*{Additional Definitions}

We make use of the matrices defined in \Cref{sec:flow}, except now these matrices are defined with respect to the modified graph.
In particular, $K$, $L$, and $R$ are now matrices with dimensions $\mnew\times \mnew$, $\mnew\times \nnew$, and $\nnew\times \mnew$ respectively.

Define $\tilde{L}$ to be the $kn\times \nnew$ matrix obtained by vertically concatenating $L[\outedges{s},\ast]$ over all original vertices $s$.
Similarly, define $\tilde{R}$ to be the $\nnew\times kn$ matrix obtained by horizontally concatenating $R[\ast, \inedges{t}]$ over all original vertices $t$. 

\paragraph*{The Algorithm}

Using the above definitions, we present our  approach for solving \textsf{$k$-APC} in \Cref{alg:k-APC}.

\begin{algorithm}[ht]
\caption{Our algorithm for solving \textsf{$k$-APC}.}
\label{alg:k-APC}
    \begin{algorithmic}[1] 
       
        \State 
        Compute the $\nnew\times \nnew$ matrix $(I - RL)^{-1}$.
        \label{step:encode-invert}

        \State
        Compute the $k\nnew\times k\nnew$ matrix $M = \tilde{L}(I-RL)^{-1}\tilde{R}$.
        \label{step:encode-multiply}

        \State
        For each pair $(s,t)$ of original vertices, compute
            \[\rank M[\outedges{s},\inedges{t}]\]
        and output this as the value for $\min(k,\lambda(s,t))$.
        \label{step:decode}
    \end{algorithmic}
\end{algorithm}

\begin{theorem}
\label{thm:correct}
With probability at least $1-5/(\mnew)$, \Cref{alg:k-APC} correctly solves \textsf{$k$-APC}.
\end{theorem}
\begin{proof}
By \Cref{lem:invertible}, with probability at least $1-1/(\mnew)^4$ the matrix $I-K$ is invertible.
Going forward, we assume that $I-K$ is invertible.



By \Cref{lem:conn-equality}, with probability at least $1 - 4/(\mnew)^3$, we have 
    \begin{equation}
    \label{eq:bounded-flow-exact}
     \rank (I-K)^{-1}[\outedges{s}, \inedges{t}] = \min(k,\lambda(s,t))
     \end{equation}
   for any given original vertices $s$ and $t$. 
By union bound over all $n^2\le (\mnew)^2$ pairs  of original vertices $(s,t)$, we see that \Cref{eq:bounded-flow-exact} holds for all original vertices $s$ and $t$ with probability at least $ 1- 4/(\mnew)$.

Since $I-K$ is invertible, by \Cref{claim:low-rank} and \Cref{lem:inverse} we have
    \[(I-K)^{-1} = (I-LR)^{-1} = I + L(I-RL)^{-1}R.\]
Using the above equation in \cref{eq:bounded-flow-exact} shows that for original vertices $s$ and $t$, the quantity $\min(k,\lambda(s,t))$ is equal to the rank of 
    \[ (I + L(I-RL)^{-1}R)[\outedges{s}, \inedges{t}]  = L[\outedges{s}, \ast](I-RL)^{-1}R[\ast, \inedges{t}]\]
where we use the fact that $I[\outedges{s}, \inedges{t}]$ is the all zeroes matrix, since in the modified graph, $\outedges{s}$ and $\inedges{t}$ are disjoint sets for all pairs of original vertices $(s,t)$. 

Then by definition of $\tilde{L}$ and $\tilde{R}$, the above equation and discussion imply that 
\[\min(k,\lambda(s,t)) = \rank~(\tilde{L}(I-RL)^{-1}\tilde{R})[\outedges{s}, \inedges{t}] = \rank M[\outedges{s}, \inedges{t}]\]
which proves that \Cref{alg:k-APC} outputs the correct answers.

A union bound over the events that $I-K$ is invertible and that \cref{eq:bounded-flow-exact} holds for all $(s,t)$, shows that \Cref{alg:k-APC} is correct with probability at least $1 - 5/(\mnew)$.
\end{proof}

We are now ready to prove our main result.

\flowthm*
\begin{proof}

By \Cref{thm:correct},  \Cref{alg:k-APC} correctly solves the \textsf{$k$-APC} problem.
We now argue that \Cref{alg:k-APC} can be implemented to run in $\tilde{O}((kn)^\omega)$ time.

In step \ref{step:encode-invert} of \Cref{alg:k-APC}, we need to compute $(I-RL)^{-1}$.

From the definitions of $R$ and $L$, we see that to compute $RL$, it suffices to compute the products $R_iL_j$ for each pair of indices $(i,j)\in [k]^2$.
The matrix $R_iL_j$ is $\nnew\times \nnew$, and its rows and columns are indexed by vertices in the graph.
Given vertices $u$ and $v$, let $E(u,v)$ denote the set of parallel edges from $u$ to $v$. 
From the definitions of the $R_i$ and $L_j$ matrices, we see that for any vertices $u$ and $v$, we have 
    \begin{equation}
    \label{eq:parallel-sum}
    (R_iL_j)[u,v] = \sum_{e\in E(u,v)}R_i[u,e]L_j[e,v].
    \end{equation}
 As noted in \cref{sec:prelim}, for all vertices $u$ and $v$ we may assume that $|E(u,v)| \le k$.
 
For each vertex $u$, define the $k\times \outdeg(u)$ matrix $R'_{u}$, with rows indexed by $[k]$ and columns indexed by edges exiting $u$, by setting 
    \[R'_u[i,e] = R_i[u,e]\]
for all $i\in [k]$ and $e\in \outedges{u}$.

Similarly, for 
each vertex $v$, define the  $\indeg(v)\times k$ matrix $L'_v$ by setting
    \[L'_v[e,j] = L_j[e,v]\]
for all $e\in \inedges{v}$ and $j\in [k]$.

Finally, for each pair $(u,v)$ of vertices, define $R'_{uv} = R'_u[\ast,E(u,v)]$ and $L'_{uv} = L'_v[E(u,v),\ast]$.
Then by \cref{eq:parallel-sum}, we have 
    \[(R_i L_j)[u,v] = R'_{uv} L'_{uv}[i,j].\]
Thus, to compute the $R_i L_j$ products, it suffices to build the $R'_{u}$ and $L'_v$ matrices in $O\grp{k\mnew}$ time, and then compute the $R'_{uv} L'_{uv}$ products.
We can do this by computing $\grp{\nnew}^2$  products of pairs of $k\times k$ matrices.
Since for every pair of vertices $(u,v)$, there are at most $k$ parallel edges from $u$ to $v$, $k\mnew \le k^2n^2$, we can compute all the $R_i L_j$ products, and hence the entire matrix $RL$, in $\tilde{O}(n^2k^\omega)$ time.

We can then compute $I-RL$ by modifying $O(kn)$ entries of $RL$.
Finally, by \Cref{lem:compute-inverse} we can compute $(I-RL)^{-1}$ in $\tilde{O}((kn)^\omega)$ time.

So overall, step \ref{step:encode-invert} of \Cref{alg:k-APC} takes $\tilde{O}((kn)^\omega)$ time.

In step \ref{step:encode-multiply} of \Cref{alg:k-APC}, we need to compute $M = \tilde{L}(I-RL)^{-1}\tilde{R}$.

Recall that $\tilde{L}$ is a $kn\times \nnew$ matrix.
By definition, each row of $\tilde{L}$ has a single nonzero entry.
Similarly, $\tilde{R}$ is an $\nnew \times kn$ matrix, with a single nonzero entry in each column.

Thus we can compute $M$, and complete step \ref{step:encode-multiply}  of \Cref{alg:k-APC} in $\tilde{O}((kn)^2)$ time.

Finally, in step \ref{step:decode} of \Cref{alg:k-APC}, we need to compute 
        \begin{equation}
        \label{rankers}
        \rank M[\outedges{s}, \inedges{t}]
        \end{equation}
for each pair of original vertices $(s,t)$ in the graph.
In the modified graph, each original vertex has indegree and outdegree $k$, so each $M[\outedges{s}, \inedges{t}]$ is a $k\times k$ matrix.
For any fixed $(s,t)$, by \Cref{lem:rank} we can compute the rank of $M[\outedges{s}, \inedges{t}]$  in $\tilde{O}(k^\omega)$ time.

So we can compute the ranks from \Cref{rankers} for all $n^2$ pairs of original vertices $(s,t)$ and complete step \ref{step:decode} of \Cref{alg:k-APC} in $\tilde{O}(k^\omega n^2)$ time. 

Thus we can solve \textsf{$k$-APC} in $\tilde{O}((kn)^\omega)$ time overall, as claimed.
\end{proof}

\section{Encoding Vertex Connectivities} 
\label{sec:flow-vertex}

Take a prime $p = \tilde{\Theta}(n^5)$.
Let $K$ be an $n\times n$ matrix, whose rows and columns are indexed by vertices of $G$.
For each pair $(u,v)$ of vertices, if $(u,v)$ is an edge in $G$, we set $K[u,v]$ to be a uniform random element of $\mathbb{F}_p$.
Otherwise, $K[u,v] = 0$.

Recall from \Cref{sec:prelim} that given a vertex $v$ in $G$, we let  $\innodescl{v} = \innodes{v}\cup\set{v}$ be the set consisting of $v$ and all in-neighbors of $v$, and $\outnodescl{v} = \outnodes{v}\cup\set{v}$ be the set consisting of $v$ and all out-neighbors of $v$.
The following proposition\footnote{The result stated here differs from a similar claim used in \cite[Section 5]{ap-bounded-mincut}. We discuss this difference, and related subtleties, in \Cref{correctingtheliterature}.} is based off ideas from \cite[Section 2]{apc} and \cite[Section 5]{ap-bounded-mincut}.
We present a complete proof of this result in \Cref{anewproof}.

\begin{restatable}{proposition}{subtleay}
    \label{node-prop:original-vc}
        For any vertices $s$ and $t$ in $G$, with probability at least $1-3/n^3$, the matrix $(I-K)$ is invertible and we have 
        \[\rank~(I-K)^{-1}[\outnodescl{s},\innodescl{t}] = \begin{cases} \nu(s,t) + 1 & \text{if }(s,t)\text{ is an edge} \\ \nu(s,t) & \text{otherwise.}\end{cases}\]
\end{restatable}

\Cref{node-prop:original-vc} shows that we can compute vertex connectivities in $G$ simply by computing ranks of certain submatrices of $(I-K)^{-1}$.
However, these submatrices could potentially be quite large, which is bad if we want to compute the vertex connectivities quickly.
To overcome this issue, we show how to decrease the size of $(I-K)^{-1}$ while still preserving relevant information about the value of $\nu(s,t)$.

\begin{lemma}
    \label{lem:conditioning}
    Let $M$ be an $a\times b$ matrix over $\mathbb{F}_p$.
    Let $\Gamma$ be a $(k+1)\times a$ matrix with uniform random entries from $\mathbb{F}_p$.
    Then with probability at least $1 - (k+1)/p$, we have 
        \[\rank \Gamma M = \min(k+1, \rank M).\]
\end{lemma}
\begin{proof}
    Since $\Gamma M$ has $k+1$ rows, $\rank (\Gamma M) \le k+1$.
    
    Similarly, since $\Gamma M$ has $M$ as a factor, $\rank (\Gamma M) \le \rank M$.
    Thus 
        \begin{equation}
        \label{eq:trivial}
        \rank \Gamma M \le \min(k+1, \rank M).
        \end{equation}
    So, it suffices to show that $\Gamma M$ has rank at least $\min(k+1,\rank M)$.

    Set $r = \min(k+1,\rank M)$.
    Then there exist subsets $S$ and $T$ of row and column indices respectively, such that $|S| = |T| = r$ and $M[S,T]$ has rank $r$.
    Now, let $U$ be an arbitrary set of $r$ rows in $\Gamma$.
    Consider the matrix $M' = (\Gamma M)[U, T]$.

    We can view each entry of $M'$ as a polynomial of degree at most $1$ in the entries of $\Gamma$.
    This means that $\det M'$ is a polynomial of degree at most $r$ in the entries of $\Gamma$.
    Moreover, if the submatrix $\Gamma[U,T] = I$ happens to be the identity matrix, then $M' = M[S,T]$.
    This implies that $\det M'$ is a nonzero polynomial in the entries of $\Gamma$, because for some assignment of values to the entries of $\Gamma$, this polynomial has nonzero evaluation $\det M[S,T]\neq 0$ (where we are using the fact that $M[S,T]$ has full rank).

    So by the Schwartz-Zippel Lemma (\Cref{schwartz-zippel}), the matrix $\Gamma M$ has rank at least $r$, with probability at least $1-r/p$.
    
    Together with \Cref{eq:trivial}, this implies the desired result.    
\end{proof}

Now, to each vertex $u$ in the graph, we assign a $(k+1)$-dimensional column vector $\vec{b}_u$ and a $(k+1)$-dimensional row vector $\vec{c}_u$.

Let $B$ be the $(k+1)\times n$ matrix formed by concatenating all of the $\vec{b}_u$ vectors horizontally, and let $C$ be the $n\times (k+1)$ matrix formed by concatenating all of the $\vec{c}_u$ vectors vertically.
For each pair of distinct vertices $(s,t)$, define the $(k+1)\times (k+1)$ matrix 
    \begin{equation}
        \label{eq:Mst-def}
        M_{s,t} = B[\ast,\outnodescl{s}]\grp{(I-K)^{-1}[\outnodescl{s},\innodescl{t}]}
        C[\innodescl{t}, \ast].
    \end{equation}
The following result is the basis of our algorithm for \textsf{$k$-APVC}.

\begin{lemma}
    \label{lem:vertex-conn-kbyk}
    For any vertices $s$ and $t$ in $G$, with probability at least $1 - 5/n^3$, we have
        \[\rank M_{s,t} = \begin{cases}\min(k+1, \nu(s,t)+1) & \text{if }(s,t)\text{ is an edge} \\ \min(k+1,\nu(s,t)) & \text{otherwise.}\end{cases}\]
\end{lemma}
\begin{proof}
Fix vertices $s$ and $t$. Then, by \Cref{node-prop:original-vc}, we have 
        \[\rank~(I-K)^{-1}[\outnodescl{s},\innodescl{t}] = \begin{cases}\nu(s,t)+1 & \text{if }(s,t)\text{ is an edge} \\ \nu(s,t) & \text{otherwise}\end{cases}\]
    with probability at least $1-3/n^3$.
    Assume the above equation holds.

    Then, by setting $\Gamma = B[\ast,\outnodescl{s}]$ and $M = (I-K)^{-1}[\outnodescl{s},\innodescl{t}]$ in \Cref{lem:conditioning}, we see that with probability at least $1-1/n^3$ we have 
        \[\rank B[\ast,\outnodescl{s}](I-K)^{-1}[\outnodescl{s},\innodes{t}] = \begin{cases}\min(k+1, \nu(s,t)+1) & \text{if }(s,t)\text{ is an edge} \\ \min(k+1,\nu(s,t)) & \text{otherwise.}\end{cases}.\]
    Assume the above equation holds.

    Finally, by setting $\Gamma = C^{\top}[\ast,\innodes{t}]$  and $M = ( B[\ast,\outnodescl{s}](I-K)^{-1}[\outnodescl{s},\innodes{t}])^{\top}$ in \Cref{lem:conditioning} and transposition, we see that with probability at least $1-1/n^3$ we have 
        \[\rank B[\ast,\outnodescl{s}]\grp{(I-K)^{-1}[\outnodescl{s},\innodes{t}]}C[\innodes{t},\ast] = \min(k+1, \nu(s,t)+1)\]
    if there is an edge from $s$ to $t$, and 
        \[\rank B[\ast,\outnodescl{s}]\grp{(I-K)^{-1}[\outnodescl{s},\innodes{t}]}C[\innodes{t},\ast] = \min(k+1, \nu(s,t))\]
    otherwise.
    So by union bound, the desired result holds with probability at least $1-5/n^3$.
\end{proof}

    

    


\section{Vertex Connectivity Algorithm}
\label{sec:k-APVC}


Let $A$ be the adjacency matrix of the graph $G$ with self-loops.
That is, $A$ is the $n\times n$ matrix whose rows and columns are indexed by vertices of $G$, and for every pair $(u,v)$ of vertices, $A[u,v] = 1$ if $v\in\outnodescl{u}$ (equivalently, $u\in\innodescl{v}$), and $A[u,v] = 0$ otherwise.


Recall the definitions of the $\vec{b}_u$ and $\vec{c}_u$ vectors, and the $K, B, C$ and $M_{s,t}$ matrices from \Cref{sec:flow-vertex}.
For each $i\in [k+1]$, let $P_i$ be the $n\times n$ diagonal matrix, with rows and columns indexed by vertices of $G$, such that $P_i[u,u] = \vec{b}_u[i]$.
Similarly, let $Q_i$ be the $n\times n$ diagonal matrix, with rows and columns indexed by vertices of $G$, such that $Q_i[u,u] = \vec{c}_u[i]$.

With these definitions, we present our approach for solving  \textsf{$k$-APVC} in \Cref{alg:k-APVC}.


\begin{algorithm}[ht]
\caption{Our algorithm for solving \textsf{$k$-APVC}.}
\label{alg:k-APVC}
    \begin{algorithmic}[1] 
        \State 
        Compute the $n \times n$ matrix $(I - K)^{-1}$.
        \label{node-step:encode-invert}
        \State
        For each pair $(i,j)\in [k+1]^2$ of indices, compute the $n\times n$ matrix 
        \label{node-step:encode-multiply}
            \[D_{ij} = A P_i(I-K)^{-1}Q_j A^\top.\]
            \vspace{-0.5cm}
        \State
        \label{node-step:decode}
        For each pair $(s,t)$ of vertices, let $F_{s,t}$ be the $(k+1)\times (k+1)$ matrix whose $(i,j)$ entry is equal to $D_{ij}[s,t]$.
            If $(s,t)$ is an edge, output $(\rank F_{s,t})-1$ as the value for $\min(k,\nu(s,t))$.
            Otherwise, output $\min(k,\rank F_{s,t})$ as the value for $\min(k,\nu(s,t))$.
    \end{algorithmic}
\end{algorithm}

The main idea of \Cref{alg:k-APVC} is to use \Cref{lem:vertex-conn-kbyk} to reduce computing $\min(k, \nu(s,t))$ for a given pair of vertices $(s,t)$ to computing the rank of a corresponding $(k+1)\times (k+1)$ matrix, $M_{s,t}$.
To make this approach efficient, we compute the entries of all $M_{s,t}$ matrices simultaneously, using a somewhat indirect argument.

\begin{theorem}
    \label{vertex-flow-correct}
    With probability at least $1 - 5/n$, \Cref{alg:k-APVC} correctly solves \textsf{$k$-APVC}.
\end{theorem}
\begin{proof}
We prove correctness of \Cref{alg:k-APVC} using the following claim.

\begin{claim}
    \label{simultaneousentries}
    For all pairs of indices $(i,j)\in [k+1]^2$ and all pairs of vertices $(s,t)$, we have 
            \[M_{s,t}[i,j] = D_{ij}[s,t],\]
    where $D_{ij}$ is the matrix computed in step \ref{node-step:encode-multiply} of \Cref{alg:k-APVC}.
\end{claim}
\begin{claimproof}
    By expanding out the expression for $D_{ij}$ from step \ref{node-step:encode-multiply} of \Cref{alg:k-APVC}, we have 
        \[D_{ij}[s,t] = \sum_{u, v} A[s,u] P_i[u,u] \grp{(I-K)^{-1}[u,v]}Q_j[v,v] A[v,t],\]
    where the sum is over all vertices $u, v$ in the graph (here, we  use the fact that $P_i$ and $Q_j$ are diagonal matrices).
    By the definitions of $A$,  the $P_i$, and the $Q_j$ matrices, we have
        \begin{equation}
        \label{eq:Dij-entries}
        D_{ij}[s,t] = \sum_{\substack{ u\in\outnodescl{s} \\ v\in\innodescl{t}}} \vec{b}_u[i] \grp{(I-K)^{-1}[u,v]} \vec{c}_v[j].
        \end{equation}
    On the other hand, the definition of $M_{s,t}$ from \Cref{eq:Mst-def} implies that
        \[M_{s,t}[i,j] = \sum_{\substack{u\in\outnodescl{s}\\ v\in\innodescl{t}}} B[i,u] \grp{(I-K)^{-1}[u,v]} C[v,j].\]
    Since $B[i,u] = \vec{b}_u[i]$ and $C[v,j] = \vec{c}_v[j]$, the above equation and \Cref{eq:Dij-entries} imply that 
        \[M_{s,t}[i,j] = D_{ij}[s,t]\]
     for all $(i,j)$ and $(s,t)$, as desired.
\end{claimproof}

By \Cref{simultaneousentries}, the matrix $F_{s,t}$ computed in step \ref{node-step:decode} of \Cref{alg:k-APVC} is equal to $M_{s,t}$.
So by \Cref{lem:vertex-conn-kbyk}, for any fixed pair $(s,t)$ of vertices we have 
        \begin{equation}
        \label{FinsteadofM}
        \rank F_{s,t} = \begin{cases}\min(k+1, \nu(s,t)+1) & \text{if }(s,t)\text{ is an edge} \\ \min(k+1,\nu(s,t)) & \text{otherwise.}\end{cases}
        \end{equation}
     with probability at least $1-5/n^3$.
     Then by a union bound over all pairs of vertices $(s,t)$, we see that \Cref{FinsteadofM} holds for all pairs $(s,t)$, with probability at least $1-5/n$.
     
     Assume this event occurs. 
     Then if $(s,t)$ is an edge, by \Cref{FinsteadofM} we correctly return 
        \[(\rank F_{s,t}) - 1 = \min(k+1, \nu(s,t)+1) - 1 = \min(k,\nu(s,t))\]
    as our answer for this pair.

    Similarly, if $(s,t)$ is not an edge, by \Cref{FinsteadofM} we correctly return 
        \[\min(k,\rank F_{s,t}) = \min(k,k+1,\nu(s,t)) = \min(k,\nu(s,t))\]
    as our answer for this pair.
This proves the desired result.
\end{proof}

With \Cref{vertex-flow-correct} established, we can prove our result for vertex connectivities.

\vertexflowthm*
\begin{proof}
    By \Cref{vertex-flow-correct}, \Cref{alg:k-APVC} correctly solves the \textsf{$k$-APVC} problem.
    We now argue that \Cref{alg:k-APVC} can be implemented to run in $\tilde{O}(k^2n^\omega)$ time.

    In step \ref{node-step:encode-invert} of \Cref{alg:k-APVC}, we need to compute $(I-K)^{-1}$.
    Since $K$ is an $n\times n$ matrix, by \Cref{lem:compute-inverse} we can complete this step in $\tilde{O}(n^\omega)$ time.

    In step \ref{node-step:encode-multiply} of \Cref{alg:k-APVC}, we need to compute $D_{ij}$ for each pair $(i,j)\in [k+1]^2$.
    For each fixed pair $(i,j)$, the $D_{ij}$ matrix is defined as a product of five $n\times n$ matrices whose entries we know, so this step takes  $\tilde{O}(k^2n^\omega)$ time overall.

    In step \ref{node-step:decode} of \Cref{alg:k-APVC}, we need to construct each $F_{st}$ matrix, and compute its rank.
    Since each $F_{st}$ matrix has dimensions $(k+1)\times (k+1)$ and its entries can be filled in simply by reading entries of the $D_{ij}$ matrices we have already computed, by \Cref{lem:rank} this step can be completed in $\tilde{O}(k^{\omega}n^2)$ time.

    By adding up the runtimes for each of the steps and noting that $k\le n$, we see that \Cref{alg:k-APVC} solves \textsf{$k$-APVC} in $\tilde{O}(k^2n^\omega)$ time, as claimed.
\end{proof}

\section{Conclusion}
\label{sec:conclusion}
 
In this paper, we presented algorithms solving \textsf{$k$-APC} and \textsf{$k$-APVC} in  $\tilde{O}((kn)^\omega)$ and $\tilde{O}(k^2n^\omega)$ time respectively.
Many open problems remain concerning the exact time complexity of these problems.
We highlight some open questions we find particularly interesting:

\begin{enumerate}
    \item 
    The most relevant question to our work: can we solve \textsf{$k$-APC} or \textsf{$k$-APVC} faster? 
    Is it possible to solve \textsf{$k$-APC} in $\tilde{O}(k^2n^\omega)$ time, as fast as our algorithm for \textsf{$k$-APVC}? 
    Could there be some moderately large parameter values $k\ge n^{\Omega(1)}$ for which these problems can be solved in $\tilde{O}(n^\omega)$ time, matching the runtime for constant $k$? 

    \item
    Can we get better conditional lower bounds for \textsf{$k$-APC} and \textsf{$k$-APVC}?
    Currently, no conditional lower bound rules out the possibility that these problems could, for example, be solved in 
    $\tilde O(k n^\omega)$ time.
    For the \textsf{APC} and \textsf{APVC} problems, can the known $n^{\omega(1,2,1)-o(1)}$ conditional lower bounds 
    be improved\footnote{There is some evidence that better lower bounds may be difficult to establish \cite{Trabelsi23}.} to $n^{4-o(1)}$ conditional lower bounds?

    \item 
    Recently, \cite{Trabelsi23} showed that there is a nondeterministic verifier for the \textsf{APVC} problem, running in $O(n^{\omega(1,2,1)})$ time.
    Is there a nondeterministic verifier for \textsf{APC} with the same runtime? 
    Are there nondeterministic verifiers for the \textsf{$k$-APC} and \textsf{$k$-APVC} problems, which run faster than the algorithms from \Cref{thm:flow,thm:vertex-flow}?


\end{enumerate}

\bibliography{main}

\appendix

\section{Conjectures in Fine-Grained Complexity}
\label{sec:fg-conjectures}

In the \textsf{Boolean Matrix Multiplication (BMM)} problem, we are given $n\times n$ matrices $A$ and $B$ with entries in $\set{0,1}$, and are tasked with computing, for each pair $(i,j)\in [n]^2$, whether there exists an index $k$ such that $A[i,k] = B[k,j] = 1$.
Using matrix multiplication, we can solve \textsf{BMM} in $O(n^\omega)$ time.

The \textsf{BMM} hypothesis\footnote{The literature also refers to the \textsf{Combinatorial BMM} hypothesis, an informal conjecture that no ``combinatorial'' algorithm for \textsf{BMM} runs in $O(n^{3-\delta})$ time, for any constant $\delta > 0$.} posits that this is essentially optimal, and asserts that there is no constant $\delta > 0$ such that \textsf{BMM} can be solved in $O(n^{\omega - \delta})$ time. 

Let $k$ be a positive integer.
In the \textsf{$k$SAT} problem, we are given a \textsf{$k$-CNF} (a Boolean formula which can be written as a conjunction of clauses, where each clause is the disjunction of at most $k$ Boolean literals) over $n$ variables, and tasked with determining if there is some assignment of values to the variables which satisfies the \textsf{$k$-CNF}.

The \textsf{Strong Exponential Time Hypothesis (SETH)} conjectures that for any constant $\delta > 0$, there exists some positive integer $k$ such that \textsf{$k$SAT} cannot be solved in $2^{(1-\delta)n}\poly(n)$ time. 

In the \textsf{$4$-Clique} problem, we are given a graph $G$, and tasked with determining if it contains a clique on four vertices (i.e., four distinct vertices which are mutually adjacent).

Let $\omega(1,2,1)$ be the smallest positive real such that we can multiply an $n\times n^2$ matrix with an $n^2\times n$ matrix in $n^{\omega(1,2,1) + o(1)}$ time.
It is known that \textsf{$4$-Clique} can be solved in $O\grp{n^{\omega(1,2,1)}}$ time. 
The \textsf{$4$-Clique Conjecture\footnote{This conjecture also has informal counterpart in the literature, which states that no ``combinatorial'' algorithm for \textsf{$4$-Clique} runs in $O(n^{4-\delta})$ time, for any constant $\delta > 0$.}} asserts that this runtime is essentially optimal, in the sense that for any constant $\delta > 0$, the \textsf{$4$-Clique} problem cannot be solved in $O\grp{n^{\omega(1,2,1)-\delta}}$ time.

\section{Vertex Connectivity Encoding Scheme}
\label{sec:app-proof}

In \Cref{correctingtheliterature}, we describe the result that \cite{ap-bounded-mincut} obtains for computing vertex cuts, and explain how the claims in \cite{ap-bounded-mincut} differ from the arguments in this work.
In \Cref{anewproof}, we present a proof of \Cref{node-prop:original-vc}.

\subsection{Discussion of Previous Vertex Cuts Algorithm}
\label{correctingtheliterature}

In \cite{ap-bounded-mincut}, the authors present an $\tilde{O}((kn)^\omega)$ time algorithm for the  \textsf{$k$-Bounded All-Pairs Minimum Vertex Cut} problem.
In this problem, we are given a directed graph $G$ on $n$ vertices, and are tasked with returning, for every pair of vertices $(s,t)$, the size of a minimum $(s,t)$-vertex cut, if this size is less than $k$.
For pairs $(s,t)$ where the size of a minimum $(s,t)$-vertex cut is at least $k$, we simply need to return the value $k$.

When $(s,t)$ is not an edge, Menger's theorem implies that the size of a minimum $(s,t)$-vertex cut is equal to the maximum number of internally vertex-disjoint paths from $s$ to $t$.
So, for such pairs $(s,t)$, the \textsf{$k$-Bounded All-Pairs Minimum Vertex Cut} problem simply requires we compute the value of $\min(k,\nu(s,t))$, as in the \textsf{$k$-APVC} problem.

However, when $(s,t)$ is an edge, as discussed in \Cref{sec:prelim}, no $(s,t)$-vertex cut can exist.
This is because no matter which vertices outside of $s$ and $t$ we delete, the resulting graph will always have a path of length one from $s$ to $t$.

In this case, it may be reasonable to define the ``size of a minimum $(s,t)$-vertex cut'' to be $\infty$.
With this convention, for pairs of vertices $(s,t)$ which are edges, we simply need to return the value of $k$ in the  \textsf{$k$-Bounded All-Pairs Minimum Vertex Cut} problem.
This is precisely what the algorithm of \cite{ap-bounded-mincut} does.

In more detail, the argument sketched 
in \cite[Proof of Lemma 5.1]{ap-bounded-mincut} argues that \begin{equation}\label{expression}\rank~(I-K)^{-1}[\outnodes{s},\innodes{t}]\end{equation} is equal to the size of a minimum $(s,t)$-vertex cut, for all pairs of vertices $(s,t)$ which are not edges, with high probability (where $K$ is defined as in \Cref{subsec:apvc}).

            The algorithm from \cite[Proof of Theorem 5.2]{ap-bounded-mincut} first modifies the graph, and then computes the value of \Cref{expression} for every pair of original vertices $(s,t)$ with respect to the new graph, to compute the answers to the  \textsf{$k$-Bounded All-Pairs Minimum Vertex Cut} problem.

    As described in \Cref{subsec:apvc}, the graph is transformed by  introducing for each vertex $s$ a set $S$ of $k$ new nodes, adding edges from $s$ to each node in $S$, adding edges from all nodes in $S$ to all out-neighbors of $s$ in the original graph, and erasing all edges from $s$ to its original out-neighbors. 
    Similarly, the transformation also introduces for each vertex $t$ a set $T$ of $k$ new nodes, adds edges from all in-neighbors of $t$ to the nodes in $T$, adds edges from all nodes in $T$ to $t$, and erases all edges entering $t$ from its original in-neighbors. 

    If $(s,t)$ was not an edge in the original graph, then $(s,t)$ is still not an edge in the new graph.
    In this scenario, the vertex connectivity from $s$ to $t$ in the new graph turns out to be $\min(k,\nu(s,t))$.
    Since $(s,t)$ is not an edge, the value $\nu(s,t)$ coincides with the size of a minimum $(s,t)$-vertex cut.
    So in this case, returning the value of \Cref{expression} produces the correct answer for the \textsf{$k$-Bounded All-Pairs Minimum Vertex Cut} problem.

    Suppose now that $(s,t)$ is an edge in the original graph.
    Then $(s,t)$ will not be an edge in the new graph, since the new out-neighbors of $s$ are all distinct from the new in-neighbors of $t$.
    However, the transformation ensures that in the new graph there are $k$ internally vertex-disjoint paths from $s$ to $t$, because each of the $k$ new out-neighbors of $s$ has an edge to each of the $k$ new in-neighbors of $t$. 
    Then in this case, returning the value of \Cref{expression} just amounts to returning the value $k$.
    
    So the algorithm of \cite[Proof of Theorem 5.2]{ap-bounded-mincut} produces the correct answer for the \textsf{$k$-Bounded All-Pairs Minimum Vertex Cut} problem when $(s,t)$ is an edge, using the convention that the size of a minimum $(s,t)$-vertex cut is $\infty$ when $(s,t)$ is an edge.

In summary: the algorithm from \cite{ap-bounded-mincut} runs in $\tilde{O}((kn)^\omega)$ time and computes $\min(k,\nu(s,t))$ for all pairs of vertices $(s,t)$ such that $(s,t)$ is not an edge, but for pairs where $(s,t)$ is an edge, does not return any information about the value of $\nu(s,t)$.

\paragraph*{Moving From Vertex Cuts to Vertex Connectivities}

        The quantity in \Cref{expression} differs from the expression we use  \Cref{node-prop:original-vc}, where we index the rows and columns by $\outnodescl{s}$ and $\innodescl{t}$, instead of $\outnodes{s}$ and $\innodes{t}$ respectively.
    For the purpose of solving \textsf{$k$-APVC}, we need to work with a different submatrix from the one used in \Cref{expression}, because 
    when $(s,t)$ is an edge, the expression in \Cref{expression} is not necessarily equal to $\nu(s,t)$.
    
    For example, suppose $G$ is a graph where $\indeg(s) = 0$, $\outdeg(t)=1$, there is an edge from $s$ to $t$, and $\nu(s,t) = 1$.
    Then the proof sketch in \cite[Proof of Lemma 5.1]{ap-bounded-mincut} (using the terminology of \Cref{subsec:apvc}) suggests pumping unit vectors to nodes in $\outnodes{s}$, using these initial vectors to determine flow vectors for all nodes, and then computing the rank of the flow vectors assigned to nodes in $\innodes{t}$.
    In this example, since $s$ has indegree zero and no initial vector is pumped to it, it is assigned the flow vector $\vec{s} = \vec{0}$.
    Since $\innodes{t} = \set{s}$ in this example, the rank of flow vectors in $\innodes{t}$ is just zero, even though $\nu(s,t) = 1$.
    
        This issue arises more generally in the proof suggested by \cite[Proof of Lemma 5.1]{ap-bounded-mincut} whenever $(s,t)$ is an edge of the graph.
    Intuitively, this is because a maximum collection of internally vertex-disjoint paths from $s$ to $t$ will always include a path consisting of a single edge from $s$ to $t$, but if we do not pump out a vector at $s$, this path will not contribute to the rank of the flow vectors entering $t$.

    To overcome this issue and correctly compute for vertex connectivities for all pairs, we  modify the expression from \Cref{expression} appropriately (by pumping out an initial vector at the source $s$, and allowing the flow vector assigned to $t$ to contribute to teh rank), which is why our statement of \Cref{node-prop:original-vc} involves computing the rank of a different submatrix.
    
    The issue described above does not appear when dealing with edge connectivity, essentially because in that case there is always a set of edges whose removal destroys all $s$ to $t$ paths.


\subsection{Proof of Vertex Connectivity Encoding}
\label{anewproof}

In this section, we  provide a proof of \Cref{node-prop:original-vc}, whose statement we recall below.

\subtleay*

Our proof of \Cref{node-prop:original-vc} is similar to the outline in \cite[Section 5]{ap-bounded-mincut} and is based off ideas from \cite[Section 2]{apc}, but involves several modifications needed to address issues which arise when dealing with vertex connectivities. 

Fix a source node $s$ and target node $t$ in the input graph $G$.
Write $v_0 = s$.
Suppose $s$ has outdegree $d = \outdeg(s)$.
Let $v_1, \dots, v_d$ be the out-neighbors of $s$ from $\outnodes{s}$.

Take a prime $p = \Theta(n^5)$ (this is the same prime $p$ from \Cref{sec:flow-vertex}).
Let $\vec{u}_0, \dots, \vec{u}_d$ denote distinct unit vectors in $\mathbb{F}_p^{d+1}$.

Eventually, we will assign each vertex $v$ in $G$ a vector $\vec{v}\in\mathbb{F}_p^{d+1}$, which we call a \emph{flow vector}.
These flow vectors are determined by a system of vector equations.
To describe these equations, we first need to introduce some symbolic matrices.

Let $Z$ be an $n\times n$ matrix, whose rows and columns are indexed by vertices in $G$. 

If $(u,v)$ is an edge, we set
    $Z[u,v] = z_{uv}$
to be an indeterminate, and otherwise $Z[u,v] = 0$. 

Consider the following procedure.
We assign independent, uniform random values from the field $\mathbb{F}_p$ to each variable $z_{uv}$.
Let $K$ be the $n\times n$ matrix resulting from this assignment to the variables in $Z$ (this agrees with the definition of $K$ in \Cref{sec:flow-vertex}).

Now, to each vertex $v$, we assign a flow vector $\vec{v}$, satisfying the following equalities:
\begin{enumerate}
    \item 
        Recall that $\outnodescl{s} = \set{v_0, \dots, v_d}$.
        For each vertex $v_i$, we require its flow vector satisfy
            \begin{equation}
            \label{vertex-flow-sout}
            \vec{v}_i = \grp{\sum_{u\in \innodes{v_i} } \vec{u}\cdot K[u,v_i]} + \vec{u}_i.
            \end{equation}
    \item
        For each vertex $v\not\in \outnodescl{s}$, we require its flow vector satisfy
        \begin{equation}
            \label{vertex-flow-general}
            \vec{v} = \sum_{u\in\innodes{v} } \vec{u}\cdot K[u,v].
        \end{equation}
\end{enumerate}

It is not obvious that such flow vectors exist, but we show below that they do
(with high probability over the random assignment to the $z_{uv}$ variables). 
Let $H_s$ denote the $(d+1)\times n$ matrix whose columns are indexed by vertices in $G$, such that the column associated with $v_i$ is $\vec{v}_i$ for each index $0\le i\le d$, and the rest of the columns are zero vectors.
Let $F$ be the $(d+1)\times n$ matrix, with columns indexed by vertices in $G$, such that the column associated with a vertex $v$ is the corresponding flow vector $\vec{v}$.

Then \Cref{vertex-flow-sout,vertex-flow-general} are captured by the matrix equation
    \begin{equation}
    \label{node-eq:flow-matrix-equation}
    F = FK + H_s.
    \end{equation}

We will prove that flow vectors $\vec{v}$ satisfying the above conditions exist, by showing that we can solve for $F$ in \Cref{node-eq:flow-matrix-equation}.
To do this, we use the following lemma.

\begin{lemma}
    \label{node-lem:invertible}
    We have $\det(I-K)\neq 0$, with probability at least $1-1/n^3$.
\end{lemma}
\begin{proof}
    Each entry of $Z$ is a polynomial of degree at most one with constant term zero.
    Since the input graph has no self-loops, the diagonal entries of $Z$ are all zeros.

    Thus, $\det(I-Z)$ is a polynomial of degree at most $n$ and constant term $1$ (which means this polynomial is nonzero).
    So by the Schwartz-Zippel Lemma (\Cref{schwartz-zippel}), $\det(I-K)$ is nonzero with probability at least $1 - n/p \ge 1/n^4$ (by taking $p\ge n^5)$ as claimed.
\end{proof}
Suppose from now on that $\det(I-K)\neq 0$ (by \Cref{node-lem:invertible}, this occurs with high probability).
In this case, the flow vectors are well-defined, and by \Cref{node-eq:flow-matrix-equation} occur as columns of
    \begin{equation}
    \label{node-eq:flow-closed-form}
        F = H_s(I-K)^{-1} = \frac{H_s\grp{\text{adj}(I-K)}}{\det(I-K)}.
    \end{equation}

\begin{lemma}
    \label{node-lem:an-upper-bound}
    For any vertex $t$ in $G$, with probability at least $1-1/n^3$, we have 
        \[\rank F[\ast, \innodescl{t}] \le 
        \begin{cases}\nu(s,t) + 1 & \text{if }(s,t)\text{ is an edge} \\ \nu(s,t) & \text{otherwise}.\end{cases}\]
\end{lemma}
\begin{proof}
    Abbreviate $\nu = \nu(s,t)$.
    Conceptually, this proof works by arguing that the flow vectors assigned to all in-neighbors of $t$ are linear combinations of the flow vectors assigned to nodes in a minimum $(s,t)$-vertex cut of the graph $G$ with edge $(s,t)$ removed.

    Let $G'$ be the graph $G$ with edge $(s,t)$ removed (if $(s,t)$ is not an edge, then $G'=G$).
    Let $C'$ be a minimum $(s,t)$-vertex cut of $G'$.
    This means that $C'$ is a minimum size set of nodes (where $s,t\notin C'$) such that deleting $C'$ from $G'$ produces a graph with no $s$ to $t$ path.
    After removing the nodes in $C'$ from $G'$, let $S$ be the set of vertices reachable from $s$, and $T$ be the set of vertices which can reach $t$.
 
    If $(s,t)$ is an edge, set $C = C'\cup\set{s,t}$.
    Otherwise, set $C = C'$.
    
    By Menger's theorem, $|C|=\nu$ if $(s,t)$ is not an edge in $G$, and otherwise $|C| = \nu+1$.
    This is because in addition to the edge $(s,t)$, we can find $|C'|$ internally-vertex disjoint paths from $s$ to $t$, so $\nu(s,t) = |C'|+1=|C|-1$.

    Let $V'$ be the set of vertices which are in $T$, or have an edge to a vertex in $T$.
    Then set $K' = K[V',V']$ and $F' = F[\ast, V']$.
    Now define the matrix $H' = F' - F'K'.$
    By definition, the columns of $H'$ are indexed by vertices in $V'$.

    \begin{claim}
    \label{fixaroni}
        For all $v\in V'\setminus C$, we have $H'[\ast,v] = \vec{0}$.
    \end{claim}
    \begin{claimproof}
    Take $v\in V'\setminus C$.

    Note that $v\neq s$.
    Indeed, if $(s,t)$ is an edge, then $s\in C$, so $v\neq s$ is forced since $v\not\in C$.
    If $(s,t)$ is not an edge, then by definition of a vertex cut, $s$ can have no edge to $T$ and $s$ is not in $T$, which means that $s\not\in V'$, again forcing $v\neq s$.
    
    We first handle the case where $v\neq t$.

    \textbf{Case 1: $v\neq t$}

    In this case, by definition of a vertex cut, $v\in T$.
    Thus, the in-neighbors of $v$ are all members of $V'$.
    By the discussion at the beginning of this proof, $v\neq s$.

    We claim that $v\not\in\outnodes{s}$.
    Indeed, suppose to the contrary that $v\in \outnodes{s}$.
    By the case assumption, $v\neq t$.
    By definition, $v\not\in C$ and $v$ has an edge to $T$.
    This means there is a path of length two from $s$ to a vertex in $T$, not using any vertices in $C$.
    This contradicts the definition of an $(s,t)$-vertex cut.
    So it must be the case that $v\not\in\outnodes{s}$ as claimed.

    So in this case, we have shown that for $v\in V'\setminus C$, we have $\innodes{v} \sub V'$ and $v\not\in\outnodes{s}$.
        It follows from 
    \Cref{vertex-flow-general} and the definitions of $F'$ and $K'$ that 
        \begin{equation}
        \label{tedious-check}
        (F'K')[\ast,v] = \sum_{u\in \innodes{v}\cap V'} \vec{u}\cdot K'[u,v] = \sum_{u\in \innodes{v}} \vec{u}\cdot K[u,v] = F[\ast, v] = F'[\ast, v]
        \end{equation}
    which proves that 
        \[H'[\ast, v] =   F'[\ast, v] -  (F'K')[\ast,v]= \vec{0}\]
            as desired. 
    It remains to handle the case where $v\in V'\setminus C$ has $v = t$.

    \textbf{Case 2: $v = t$}
    
    In this case, by definition of $C$, there must be no edge from $s$ to $t$.

    Hence $t\not\in\outnodes{s}$.
    Of course we have $\innodes{t}\sub V'$ by definition.
    So the calculation from \Cref{tedious-check} applies to $v=t$ as well, which means that $H'[\ast, t] = \vec{0}$.

    Thus for all $v\in V'\setminus C$, we have $H'[\ast,v] = \vec{0}$ as desired.
    \end{claimproof}

    By \Cref{fixaroni}, $H'$ has at most $|C|$ nonzero columns.
    Thus, $\rank H' \le |C|$.

    Similar reasoning to the proof of \Cref{node-lem:invertible}, shows that matrix $I-K'$ is invertible, with probability at least $1-1/n^3$.
    So we can solve for $F'$ in the equation
        $H' = F'-F'K'$
    to get 
        \[F' = H'(I-K')^{-1}.\]
    
    Since $\rank H'\le |C|$, the above equation implies that $\rank F'\le |C|$ as well.
    
    By definition, $\innodescl{t}\sub V'$, so $F[\ast, \innodescl{t}]$ is a submatrix of $F'$. It follows that 
        \[\rank F[\ast, \innodescl{t}] \le |C|.\]
        
    Since $|C| = \nu+1$ if $(s,t)$ is an edge, and otherwise $|C|=\nu$, the desired result follows.   
\end{proof}

\begin{lemma}
    \label{node-lem:a-lower-bound}
    For any vertex $t$ in $G$, with probability at least $1-2/n^3$, we have 
        \[\rank F[\ast, \innodescl{t}] \ge \begin{cases}\nu(s,t) + 1 & \text{if }(s,t)\text{ is an edge} \\ \nu(s,t) & \text{otherwise}.\end{cases}\]
\end{lemma}
\begin{proof}
    Abbreviate $\nu = \nu(s,t)$.
    Intuitively, our proof will argue that the presence of internally vertex-disjoint paths from $s$ to $t$ will lead to certain nodes in $\innodescl{t}$ being assigned linearly independent flow vectors (with high probability), which then implies the desired lower bound.
    
    By definition, $G$ contains $\nu$ internally-vertex disjoint paths $P_1, \dots, P_{\nu}$ from $s$ to $t$.
    
    Consider the following assignment to the $z_{uv}$ variables of the matrix $Z$.
    For each $(u,v)$, we set $z_{uv} = 1$ if $(u,v)$ is an edge on one of the $P_i$ paths.
    Otherwise, we set $z_{uv} = 0$.
    For each vertex $v$, let $\vec{v}$ denote the flow vector for $v$ with respect to this assignment.
    To show the desired rank lower bound, it will help to first identify the values of some of these flow vectors.

    Let $\vec{u}_0$ denote the unit vector initially pumped out at vertex $s$ (as in \Cref{vertex-flow-sout}).

    Take an arbitrary $P_i$ path.
    
    First, suppose that $P_i$ is a path of length at least two.
    In this case, let $a_i$ denote the second vertex in $P_i$, and $b_i$ denote the penultimate vertex in $P_i$ (note that if $P_i$ has length exactly two, then $a_i = b_i$).
    By definition, $a_i\in\outnodes{s}$.
    Let $\vec{u}_i$ be the unit vector initially pumped at node $a_i$ (as in \Cref{vertex-flow-sout}).
    Then from our choice of assignment to the $z_{uv}$ variables, by \Cref{vertex-flow-sout} we have    
        \[\vec{a}_i = \vec{u}_0 + \vec{u}_i.\]
    Then by applying \Cref{vertex-flow-general} repeatedly to the vertices on the path from $a_i$ to $b_i$ on $P_i$, we find that
        \begin{equation}
        \label{eq:bygonesbebygones}
        \vec{b_i} = \vec{u}_0 + \vec{u}_i
        \end{equation}
    as well.
    The above equation characterizes the flow vectors for the penultimate vertices of $P_i$ paths of length at least two.
    It remains to consider the case where $P_i$ has length one.
    If $P_i$ has length one, it consists of a single edge from $s$ to $t$.
    Then by \Cref{vertex-flow-sout} we have
        \begin{equation}
        \label{amirabboud}
        \vec{s} = \vec{u}_0.
        \end{equation}

    We are now ready to show that the flow vectors for nodes in $\innodescl{t}$ together achieve the desired rank lower bound, for this particular assignment of values to the $z_{uv}$ variables.

    \begin{claim}
    \label{gottaeat}
        With respect to the assignment where $z_{uv} = 1$ if $(u,v)$ is an edge in a $P_i$ path, and $z_{uv} = 0$ otherwise, the rank of the flow vectors in $\innodescl{t}$ is at least $\nu+1$ if $(s,t)$ is an edge, and at least $\nu$ otherwise.
    \end{claim}
    \begin{claimproof}
            We perform casework on whether $(s,t)$ is an edge or not.

    \textbf{Case 1: $(s,t)$ is not an edge}

    Suppose that $(s,t)$ is not an edge.
    Then every path $P_i$ has length at least two.
\Cref{eq:bygonesbebygones} shows that the flow vectors in $\innodes{t}$ include $\vec{u}_0 + \vec{u}_1,\vec{u}_0 + \vec{u}_2,\dots,\vec{u}_0 + \vec{u}_{\nu}$. Since the $\vec{u}_i$ are distinct unit vectors, these flow vectors have rank at least $\nu$ as desired.
    
    \textbf{Case 2: $(s,t)$ is an edge}

    Suppose instead that $(s,t)$ is an edge.
    Then one of the paths in our maximum collection of vertex-disjoint paths must be a direct edge from $s$ to $t$.
    Without loss of generality, let $P_\nu$ be this path of length one (so that $P_i$ is a path of length two for all $1\le i\le \nu-1$).

    In this case, $t\in\outnodes{s}$.
    Let $\vec{u}_{\nu}$ denote the unit vector pumped out at $\vec{t}$.

    By substituting \Cref{eq:bygonesbebygones,amirabboud}
 into \Cref{vertex-flow-sout}, we get that
        \[\vec{t} = \vec{s} + \grp{\sum_{i=1}^{\nu-1} \vec{b}_i}  + \vec{u}_\nu  = \nu\cdot \vec{u}_0 + \grp{\sum_{i=1}^{\nu-1} \vec{u}_i} + \vec{u}_{\nu}.\]
        
    By combining the above equation with  \Cref{eq:bygonesbebygones,amirabboud}, we see that flow vectors assigned to nodes in $\innodescl{t}$, which include $\vec{b}_1,\dots,\vec{b}_{\nu-1},\vec{s}$, and $\vec{t}$, span the space containing distinct unit vectors $\vec{u}_0,\vec{u}_1,\dots,\vec{u}_\nu$, and hence have rank at least $\nu + 1$ as desired.
    \end{claimproof}

    For convenience, in the remainder of this proof, we let $\tilde{\nu}$ denote $\tilde{\nu} = \nu+1$ if $(s,t)$ is an edge, and set $\tilde{\nu}=\nu$ otherwise.

    By \Cref{gottaeat}, $\rank F[\ast, \innodes{t}] = \tilde{\nu}$.
    So, $F[\ast, \innodes{t}]$ contains a full rank $\tilde{\nu}\times\tilde{\nu}$ submatrix.
    
    Let $F'$ be such a submatrix.

    Now, before assigning values to the $z_{uv}$ variables, each entry of $\adj(I-Z)$ is a polynomial of degree at most $n$.
    So by \Cref{node-eq:flow-closed-form}, $\det F'$ is equal to some polynomial $P$ of degree at most $n\tilde \nu$, divided by the polynomial $(\det(I-Z))^{\tilde \nu}$.
    Note that $\det(I-Z)$ is nonzero polynomial, since it has constant term equal to $1$.
    
    By the definition of $F'$, we know that under the assignment of values to the variables from the statement of \Cref{gottaeat}, we have $\det(F')\neq 0$.
    Since as a polynomial 
    \begin{equation}
    \label{mysteriesofscience}
    \det(F') = P/\grp{\det(I-Z)}^{\tilde \nu},
    \end{equation}
    it must be the case that $P$ is a nonzero polynomial (because if $P$ was the zero polynomial, then $\det(F')$ would evaluate to zero under every possible assignment).

    By \Cref{node-lem:invertible}, with probability at least $1-1/n^3$, the determinant $\det(I-Z)\neq 0$ will be nonzero under a uniform random assignment to the $z_{uv}$ variables.
    Assuming this event occurs, by the Schwartz-Zippel Lemma (\Cref{schwartz-zippel}), a uniform random evaluation to the variables will additionally have $P \neq 0$ with probability at least 
        \[1 - (2\tilde \nu n)/p \ge 1 - 1/n^3\]
    by setting $p\ge 2n^{5}$.
    
    So by union bound, a uniform random assignment of values from $\mathbb{F}_p$ to the $z_{uv}$ variables will make $P$ and $\det(I-Z)$ nonzero simultaneously with probability at least $1-2/n^3$.
    
    If this happens, by \Cref{mysteriesofscience}, we have $\det(F')\neq 0$.
    This means a particular $\tilde{\nu}\times \tilde{\nu}$ submatrix of $F[\ast,\innodes{t}]$ will be full rank with the claimed probability, which proves the desired result.
\end{proof}

With these lemmas established, we can immediately prove the main result of this section.

\begin{proof}[Proof of \Cref{node-prop:original-vc}]
    Fix a pair of distinct vertices $(s,t)$.

    Substituting the definition of $H_s$ into \Cref{node-eq:flow-closed-form} implies that 
        \[F[\ast,\innodescl{t}] = (I-K)^{-1}[\outnodescl{s},\innodescl{t}].\]
    By union bound over \Cref{node-lem:an-upper-bound,node-lem:a-lower-bound}, we see that
        \[\rank~(I-K)^{-1}[\outnodescl{s},\innodescl{t}] = \rank F[\ast,\innodescl{t}] = \begin{cases}\nu(s,t)+1 & \text{if }(s,t)\text{ is an edge} \\ \nu(s,t) & \text{otherwise}\end{cases}\]
     with probability at least $1-3/n^3$.
\end{proof}

\end{document}